\documentclass[%
preprint,
 superscriptaddress,
nofootinbib,
 amsmath,amssymb,
 aps,
]{revtex4-1}
\usepackage{lineno}
\usepackage{lipsum}

\usepackage{graphicx}% Include figure files
\usepackage{dcolumn}% Align table columns on decimal point
\usepackage{bm}% bold math
\usepackage{physics}
\usepackage{enumerate}
\newcommand{\sgn}{\mathrm{sgn}\ }
\usepackage{amsthm}
\usepackage{comment}
 \newtheorem{thm}{Theorem}
 \newtheorem{defi}{Definition}
 \newtheorem{prop}{Proposition}
 
 \newtheorem{col}{Corollary}
 \usepackage{color}
 \usepackage{ascmac}

  \newcommand{\tcr}{\textcolor{black}}
  
\begin{document}
\count\footins = 1000
\preprint{APS/123-QED}

%\title{Stoichiometric Rays:\\A simple method to compute the controllable set of enzymatic reaction systems}% Force line breaks with \\
\title{A theoretical basis for cell deaths}% Force line breaks with \\

\author{Yusuke Himeoka}
\email[]{yhimeoka@g.ecc.u-tokyo.ac.jp}
 \affiliation{Universal Biology Institute, University of Tokyo, 7-3-1 Hongo, Bunkyo-ku, Tokyo, 113-0033, Japan}%Lines break automatically or can be forced with \\

 \author{Shuhei A. Horiguchi}%
 \affiliation{%
 Department of Mathematical Informatics, Graduate School of Information Science and
Technology, The University of Tokyo, 7-3-1, Hongo, Bunkyo-ku, Tokyo 113-8656, Japan
}%defi
\affiliation{
    Nano Life Science Institute, Kanazawa University, Kakumamachi, Kanazawa, 920-1192, Japan
}
\affiliation{%
Institute of Industrial Science, The University of Tokyo, 4-6-1, Komaba, Meguro-ku, Tokyo
153-8505, Japan
}%


\author{Tetsuya J. Kobayashi}
\affiliation{Universal Biology Institute, University of Tokyo, 7-3-1 Hongo, Bunkyo-ku, Tokyo, 113-0033, Japan}%Lines break automatically or can be forced with \\

\affiliation{%
 Department of Mathematical Informatics, Graduate School of Information Science and
Technology, The University of Tokyo, 7-3-1, Hongo, Bunkyo-ku, Tokyo 113-8656, Japan
}%defi

\affiliation{%
 Institute of Industrial Science, The University of Tokyo, 4-6-1, Komaba, Meguro-ku, Tokyo
153-8505, Japan
}%

\date{\today}% It is always \today, today,

\begin{abstract}
Understanding deaths and life-death boundaries of cells is a fundamental challenge in biological sciences. In this study, we present a theoretical framework for investigating cell death. We conceptualize cell death as a controllability problem within dynamical systems, and compute the life-death boundary through the development of ``stoichiometric rays''. This method utilizes enzyme activity as control parameters, exploiting the inherent property of enzymes to enhance reaction rates without affecting thermodynamic potentials. This approach facilitates the efficient evaluation of the global controllability of models. We demonstrate the utility of our framework using its application to a toy metabolic model, where we delineate the life-death boundary. The formulation of cell death through mathematical principles provides a foundation for the theoretical study of cellular mortality.

\end{abstract}

\maketitle

\section*{Significance Statement}
What is death? This fundamental question in biology lacks a clear theoretical framework despite numerous experimental studies. In this study, we present a new way to understand cell death by looking at how cells can or cannot control their states. We define a ``dead state'' as a state from which a cell cannot return to being alive. Our method, called ``Stoichiometric Rays'', helps determine if a cell's state is dead based on enzymatic reactions. By using this method, we can quantify the life-death boundary in metabolic models. The present framework provides a theoretical basis and a tool for understanding cell death.

\section{Introduction}
Death is one of the most fundamental phenomena, and comprehension of the life-death boundary is a pivotal issue in the study of biological systems. Cell death within simple unicellular model organisms such as {\it Escherichia coli} and yeast has been extensively studied \cite{Casolari2018-zo,Nakaoka2017-qy,Spivey2017-no,Wang2010-td,Allocati2015-oo,Horowitz2010-rd,Fagerlind2012-ev,Himeoka2020-ho,Schink2019-dd,Biselli2020-fs,Maire2020-ty,Wu2024-vz}. 

What do we call ``death'' and from what features of a given cell do we judge the cell is dead? Despite extensive studies, such characterization of cell death is still under debate \cite{Maire2020-ty,Schink2021-vp,Gray2019-kr,Umetani2021-vz,Laman_Trip2022-so,Walker2023-fd}. \tcr{Indeed, there is a debate as to whether the viable but non-culturable (VBNC) state \cite{Oliver2010-dp,xu1982survival} is a dead state or not, without a clear definition of death \cite{Song2021-xv,Kirschner2021-az}. Confusion due to lack of a clear definition has also occurred in research on bacterial persistence \cite{Jung2019-xl,Amato2014-xi,Balaban2011-rx,Balaban2004-da} and the community has reached a consensus on the definition of persister \cite{Balaban2019-zp} to avoid unnecessary slowing down of research. A clear definition of death is essential for the advancement of cell death research.} 

%Although experimental researches have explored this phenomenon in depth, theoretical frameworks for defining and characterizing cell death are conspicuously underdeveloped.

The current circumstances of microbial cell death research are in stark contrast to other topics in systems biology, where both experimental and theoretical approaches are integrated to explore cellular phenomena quantitatively and to elucidate the underlying biochemical design principles. Microbial adaptation, robustness of living systems, and network motifs are hallmarks of crosstalk between experimental and theoretical approaches \cite{Barkai1997-by,Berg2004-aa,Shinar2007-vj,Shinar2010-ns,Hopfield1974-tj}. However, despite substantial experimental data, theoretical explorations of cell death remain underdeveloped. Theoretical studies on microbial death have primarily focused on estimating death rates and assessing their impacts on population dynamics supposing a predefined concept of ``death'' \cite{Biselli2020-fs,Allocati2015-oo,Horowitz2010-rd,Fagerlind2012-ev,Himeoka2020-ho}. Theories for delineating what is death and evaluating cellular viability are indispensable for extracting the quantitative nature of cell death, and accordingly, life. 

The necessity for theories of cell death is beyond comprehending experimental observations; it is imperative for the advancement of biological sciences, particularly through the integration of burgeoning computational technologies. Recent developments in computational biology have facilitated the construction of whole-scale kinetic cell models \cite{Thornburg2022-nm} and the application of advanced machine-learning technologies, enabling large-scale simulations of cellular responses to environmental and genetic changes \cite{Li2022-cq,Choudhury2022-vv}. This includes modeling scenarios for cell death as well as cell growth \cite{Himeoka2022-dh,Thornburg2022-nm}. The development of a robust mathematical framework to define and evaluate cell death within these models is essential to advance our understanding of cellular mortality.

In the present manuscript, we aimed at constructing a theoretical foundation for cell death using cell models described by ordinary differential equations. We introduce the definition of the dead state based on the controllability of the cellular state to predefined living states. Next, we develop a tool for judging whether a given state is dead for enzymatic reaction models, the {\it stoichiometric rays}. The method leverages the inherent feature of enzymatic reactions in which the catalyst enhances the reaction rate without altering the equilibrium. The stoichiometric rays enables the efficient computation of the global controllability of the models.

\section{A framework for cell deaths} \label{sec:framework}
%In our framework, we suppose that the {\it representative living state(s)} are defined a priori, and the dead and dying states are states from which we cannot control the state back to any representative living state. Here, the control is defined as the temporal modulation of several parameters in the reaction rate function, i.e., allowing the parameter to temporally vary, $\bm u=\bm u (t)$.

In the present manuscript, we formulate the death judgment problem as a controllability problem of mathematical models of biochemical systems. \tcr{Figure~\ref{fig:scheme} is a graphical summary of the framework proposed in the present manuscript.} Herein, we focus on well-mixed deterministic biochemical reaction systems. The dynamics of these systems are described by the following ordinary differential equations: 
\begin{equation}
    \dv{\bm x (t)}{t}= \mathbb S \bm J(\bm u(t),\bm x(t)),\label{eq:ode}
\end{equation}
where $\bm x:[0,T]\to \mathbb R_{\geq 0}^N$, $\bm u:[0,T]\to\mathbb R^P$, and $\bm J:\mathbb R^P\times\mathbb R_{\geq 0}^N\to \mathbb R^R$ represent the concentration of the chemical species, control parameters, and reaction flux, respectively. Here we denote the number of chemical species, control parameters, and reactions as $N,P$ and $R$, respectively, and $\mathbb S$ is $N\times R$ stoichiometric matrix.

\begin{figure}[h!]
    \begin{center}
    \includegraphics[width = 120 mm, angle = 0]{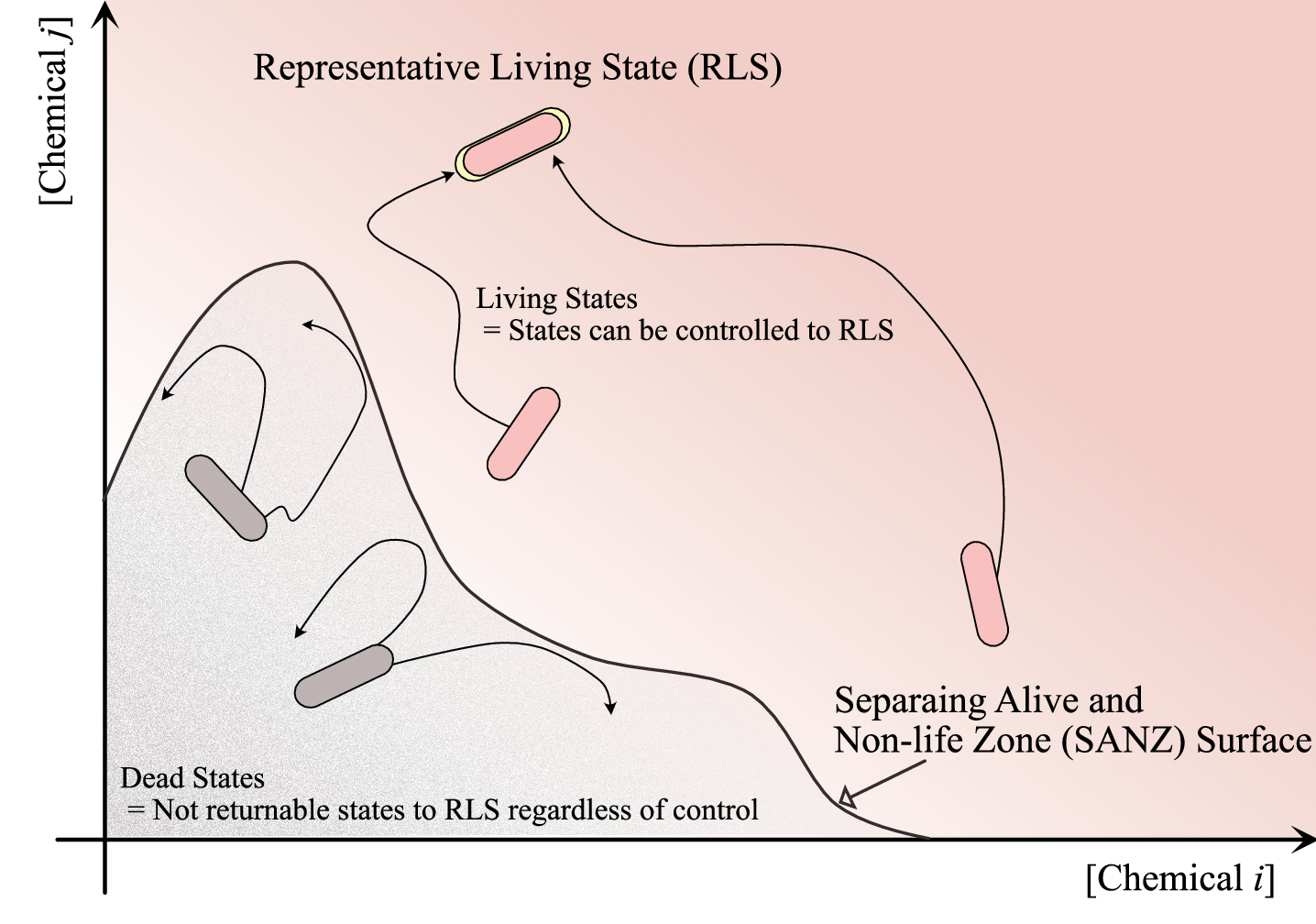}
    \caption{\tcr{A graphical summary of the proposed framework. We consider ``death'' as the loss of controllability to the predefined representative living states (RLS), $X$. If a given state can be controllabled back to the RLS through modulations of biochemical parameters (concisely defined in the main body), the state is judged to be ``living''. If no control back to the RLS is possible, regardless of the control, the state is ``dead'' in our framework. We call the set of dead states the dead set $D(X)$ as a function of the RLS. The boundary between the controllable set and the dead set $\partial C(X)\cap\partial D(X)$ is called {\it Separating Alive and Non-life Zone (SANZ) hypersurface} $\Gamma(X)$, where SANZ is taken from {\it Sanzu River}, the mythological river that separates the world of life and the afterlife in Japanese Buddhism \cite{stone2008death}. In the figure, the pink and gray colours represent the living and dead states respectively. The RLS is by definition living and is coloured pink with the yellow highlighting. The two states in the middle right are considered living and are coloured pink because they can be controlled back to RLS. On the other hand, the two states in the left middle are coloured gray because they cannot be controlled back to RLS. The existence of control back to the RLS can be computed by the proposed method {\it the stoichiometric rays} (Sec.~\ref{sec:rays}) for models where the enzymatic activities and external metabolites concentrations as the control parameters. Note that our framework does not say anything about which states should be in the RLS, while some useful mathematical arguments for the choice of the RLS are provided in the Appendix.}}
        \label{fig:scheme}
      \end{center}
    \end{figure}

\tcr{In the following, we formally introduce the definition of dead states for the further mathematical arguments in the appendix. Intuitively, we first need to define the ``representative living states'' which are reference points of the living states. The choice of the representative living state is not given by the framework. We assume that one has a working hypothesis about which states should be chosen as representative living states. Then we define the dead states as the state from which it is impossible to return to any of the representative living states, regardless of control, such as modulation of enzymatic activities and the external concentrations of metabolites.}

\begin{defi}[Trajectories]\label{def:trajectories}
We gather all the solutions of Eq.~\eqref{eq:ode} satisfying $\bm x(0)=\bm x_0$ and $\bm x(T)=\bm x_1$ and denote the set by ${\cal T}(T,\bm x_0,\bm x_1)$. In addition, the union of ${\cal T}(T,\bm x_0,\bm x_1)$ for all $T\geq 0$ is denoted as 
\begin{equation}
    {\cal T}(\bm x_0,\bm x_1):=\bigcup_{T\in[0,\infty)}{\cal T}(T,\bm x_0,\bm x_1).
\end{equation}
\end{defi}
Note that ${\cal T}(T,\bm x_0,\bm x_1)$ and ${\cal T}(\bm x_0,\bm x_1)$ contain all the solutions obtained by varying the control parameter $\bm u(t)$.

\begin{defi}[Controllable Set]\label{def:ctrl_set}
    The controllable set of given state $\bm x$ is defined as the set of states from which the state can be controlled to the target state $\bm x$. It is given by 
    \begin{eqnarray}
        C(\bm x) &=& \{\bm x_0\in \mathbb R_{\geq 0}^N \mid {\mathcal T}(\bm x_0,\bm x)\neq \varnothing\}.
    \end{eqnarray}
    We define the controllable set of a given set $X$ as a union of the controllable set of points in $X$, that is, 
    \begin{eqnarray}
        C(X)&=&\bigcup_{\bm x\in X}C(\bm x)\\
        &=&\{\bm x_0\in \mathbb R_{\geq 0}^N \mid {}^\exists \bm x\in X\ {\mathrm s.t.}\ {\mathcal T}(\bm x_0,\bm x)\neq \varnothing\}.
    \end{eqnarray}
\end{defi}

In the present framework, we assume that the {\it representative living states} $X$ are given a priori as reference points of the living states. Using the controllable set and representative living states, we define the dead state as follows:
\begin{defi}[Dead State]\label{def:death}
    A state $\bm z$ is called the dead state with respect to the representative living states $X$ if 
    \begin{equation}
        \bm z\notin C(X),
    \end{equation}
    holds. 

    The set of all dead states with respect to $X$ form the dead set and is denoted as $D(X)$, i.e., 
    \begin{equation}
        D(X)=\mathbb R_{\geq 0}^N\backslash C(X).
    \end{equation}
        \tcr{We term the boundary of the controllable set and the dead set $\Gamma(X):=\partial C(X)\cap\partial D(X)$} as the {\it Separating Alive and Non-life Zone (SANZ) hypersurface}. 
\end{defi}
The {\it SANZ hypersurface} is derived from a mythical river in the Japanese Buddhist tradition, the Sanzu River, which separates the world of the living and the afterlife \cite{stone2008death}. 

Note that the ``death'' is defined here based sorely upon the controllability of the state. The dead state can be a point attractor, limit cycle, infinitely long relaxation process, and etc. As Def.~\ref{def:death} states nothing on the features of the dead state except the controllability, the dead states can have, for instance, metabolic activity comparable to that of the representative living states. In such cases, the ``dying state'' would be an appropriate word to describe the state, while we use the term ``dead state'' for uniformity \footnote{\tcr{An intuitive terminology is to call the state in $D(X)$ with no temporal change (${\rm d}\bm x/{\rm d}t=0$) ``dead'' (i.e. at the fixed point) and ``dying'' otherwise. However, in the present framework, whether a given steta is dead or dying depends on the specific value of the control parameter. Moreover, one can make all states in $D(X)$ to be dying (${\rm d}\bm x/{\rm d}t\neq0$) states by keep changing the control parameter, for example by setting $\bm u(t)=\bm u_0|\sin(t)|$, where $\bm u_0$ is a constant vector. Thus, in the present framework, the distinction between dying and dead is neither essential nor possible.}}. 

\tcr{In the context of microbiology, for example, the steadily growing states are a straightforward choice of representative living states, while it is not necessary to be the growing states. In this case, the definition of death (Def.~\ref{def:death}) can be understood by analogy with the plating assay. In the plating assay, cells are placed on the agar plate and are considered alive if they can form a colony within a given time period. The present definition with the growth state as the representative living state checks the potential for regrowth. In this sense, the definition is parallel to the plating assay. The differences lie in the universality of the controls that the cells can perform and the time scale of the control. In the present framework, we assume that cells are universal chemical systems that are able to find and execute a control to restart growth if there is a control to do so. We also do not care about the time scale of the control. The cells are considered to be living if they can regrow even if it takes an infinite amount of time, i.e. our ``agar plate'' never dry out. }

\tcr{
    It is also noteworthy that if the continuously growing states are chosen as representative living states, the quiescent states such as the dormant states, spore and so on should be judged as living states, since cells in these states typically retain potentials for regrowth.}

\tcr{The dead state and the dead set depend on the choice of representative living states. We have presented useful mathematical statements for choosing the representative living states in the Appendix. Readers interested in the choice of representative living states are encouraged to read that section, although it is not necessary for understanding the main body of this manuscript.}

\section{The stoichiometric rays: a tool for the death judgment}\label{sec:rays}

Thus far, we have proposed a possible definition of the dead state. This definition raises the question of whether such global controllability can be computed. In this section, we show that the controllable set is efficiently computable for models with enzymatic activity and the external concentration of chemicals as the control parameter. Before moving to the main body, we briefly describe previous studies related to the controllability of chemical reaction systems.

Most biochemical reactions can take place within reasonable timescales only by the catalytic aid of enzymes \cite{wolfenden2001depth}, and thereby, the control of intracellular states relies on the modulation of enzyme activities and concentrations (hereafter collectively referred to as activities). Thus, it is indispendable to develop a mathematical technique to compute the controllability of intracellular states by modulating enzyme activities for constructing a framework of cellular viability in terms of Def.~\ref{def:death}. 

Generic frameworks for the controllability of chemical reaction systems have been studied in the control theory field \cite{Saperstone1973-zk,Farkas1998-zk,Dochain1992-gh,Drexler2016-bu}. However, there are critical limitations in previous studies when applied to biological systems; Saperstone {\it  et al.} \cite{Saperstone1973-zk} studied the controllability of linear reaction rate functions, with which we could not model many-body reactions. In Farkas {\it  et al.} \cite{Farkas1998-zk} and Dochain {\it  et al.} \cite{Dochain1992-gh}, nonlinear models were dealt with, but only local controllability was discussed. In contrast, global controllability of a nonlinear model was studied by Drexler {\it  et al.} \cite{Drexler2016-bu}, although they allowed the kinetic rate constants to be negative-valued, to which no physical reality corresponds. 

In the following, we present a simple method to compute the global controllability of enzymatic reaction models, the {\it stoichiometric rays}. The concept is hinted at the inherent feature of enzymatic reactions; enzymes enhance reactions without altering the equilibrium \cite{atkins2023atkins}. We deal with the phase space to the positive orthant, $\mathbb R_{> 0}^N$ \footnote{In typical cases, the method works for $\mathbb R_{\geq 0}^N$, including the two examples presented in the latter half of the manuscript. However, it is known that special care is necessary for mathematical claims at the boundary ($x_i=0$ for one or more chemicals) \cite{Feinberg2019-hp}. In the present manuscript, we avoid dealing with the boundary.}.

Hereafter, we restrict our attention to the cases in which the reaction rate function in Eq,~\eqref{eq:ode} is expressed by the following form
\begin{eqnarray}
J_r(\bm u,\bm x) &=& u_r f_r(\bm x)p_r(\bm x),  \label{eq:reaction}\\
f_r(\bm x)&>&0,\\
p_r(\bm x)&=&\prod_{c} x_c^{n^+_{cr}}-k_r\prod_{c} x_c^{n^-_{cr}}. \label{eq:prop}
\end{eqnarray}
We refer to $\bm f(\bm x)$ and $\bm p(\bm x)$ as the kinetic, and thermodynamic parts, respectively. $f_r(\bm x)$ is a strictly positive function. $n_{cr}^\pm \geq 0$ represents the reaction order of chemical $c$ in the forward ($+$) or backward ($-$) reactions of the $r$th reaction. $k_r\geq 0$ corresponds to the Boltzmann factor of the reaction $r$. $\bm u\in \mathbb R_{\geq 0}^R$ is the control parameter, representing the activities of the enzymes. In the following, we focus on the modulations of enzymatic activities where the forward and backward reactions cannot be modulated independently \footnote{\tcr{The enzymes change the activation energy of reactions to facilitate the reaction rate, but cannot directly change the chemical equilibrium, i.e. the direction of the reactions \cite{atkins2023atkins}.}}. However, when the forward and backward reactions are independently modulated, for example by controlling the external concentrations of chemicals, we can extend our method to deal with such situations by treating the forward and backward reactions as independent reactions.

Note that the popular (bio)chemical reaction kinetics have the form shown in Eq,~\eqref{eq:reaction}-\eqref{eq:prop}, such as the mass-action kinetics, (generalized) Michaelis-Menten kinetics, ordered- and random many-body reaction kinetics, and the ping-pong kinetics \cite{cornish2013fundamentals}. An important feature of the reaction kinetics of this form is that the direction of the reaction is set by the thermodynamic part as $\bm p(\bm x)$ is related to the thermodynamic force of the reaction. On the other hand, the remaining part $\bm f(\bm x)$ is purely kinetic, and it modulates only the absolute value of the reaction rate, but not the direction. For reversible Michaelis-Menten kinetics, $v_{\rm max}([S]-k[P])/(1+[S]/K_S+[P]/K_P)$, the thermodynamic part corresponds to its numerator, $[S]-k[P]$, \tcr{and the remaining part $v_{\rm max}/(1+[S]/K_S+[P]/K_P)$ is the kinetic part}\footnote{\tcr{The separation of $\bm f(\bm x)$ and $\bm p(\bm x)$ is not necessary. However, the separation is useful to distinguish between the thermodynamic contribution and the kinetic contribution to the reaction rate, i.e. $\bm f$ depends on the enzymatic reaction mechanism (it is whether Michaelis-Menten, ping-pong, mass-action, etc.), whereas $\bm p$ does not. The separation also simplifies the calculation of the balance manifold, which will be introduced later.}}.

Now, we allow the enzyme activities $\bm u$ to temporally vary in $\mathbb R_{\geq 0}^R$ and compute the controllable set of a given target state $\bm x^*$, $C(\bm x^*)$. The model equation is given by %In the following, we set $T$ in Eq,~\eqref{def:ctrl_set} to unity because it is arbitrarily scaled by changing the magnitude of $\bm u(t)$. Thus, in the following, we study the controllable set (Eq,~\eqref{def:ctrl_set} with $T=1$) of the model given by
\begin{equation}
    \dv{\bm x(t)}{t}=\mathbb S\bm u(t)\odot\bm f(\bm x(t))\odot\bm p(\bm x(t)),\label{eq:std}
\end{equation}
where $\odot$ denotes the element-wise product.

First, we divide the phase space $\mathbb R^N_{> 0}$ into subsets based on the directionalities of the reactions. Note that the kinetic part is strictly positive and $\bm u\in \mathbb R_{\geq 0}^R$ holds. Consequently, the directionalities of the reactions are fully set only by the thermodynamic part $\bm p(\bm x)$ \footnote{For the directionality of the $r$th reaction at $u_r=0$, we define it by $\sgn \tilde{J}_r(\bm u,\bm x)|_{u_r=0}=\lim_{u_r\to +0}\sgn \tilde{J}_r(\bm u,\bm x)$.}. Let  $\bm \sigma$ be a binary vector in $\{1,-1\}^R$, and define the {\it direction subset} $W_{\bm \sigma}$ as
\begin{equation}
    W_{\bm \sigma}=\{\bm x\in\mathbb R_{> 0}^N\mid \sgn \bm p(\bm x)=\bm \sigma\}.
\end{equation}
Next, we introduce the {\it balance manifold} of reaction $r$ given by
\begin{equation}
{\mathcal M}_r=\{\bm x\in \mathbb R^N_{> 0}\mid p_r(\bm x)=0\}.
\end{equation}
Now the phase space $\mathbb R_{> 0}^N$ is compartmentalized by the balance manifolds $\{\mathcal M_r\}_{r=0}^{R-1}$ into the direction subsets $W_{\bm \sigma}$ \footnote{On the balance manifolds, $\sum_{c} {n^+_{cr}}\ln x_c-\ln k_r-\sum_{c}{n^-_{cr}}\ln  x_c=0$ holds. This is a linear equation of $\ln \bm x$, and thus, each direction subset is connected. Depending on how the balance manifolds intersect, the number of the direction subset changes. At most, there are $2^R$ direction subsets for a model with $R$ reactions.}.

Let {\it path} $\bm \xi(t)$ be the solution of Eq.~\eqref{eq:std} with a given $\bm u(t)$ reaching the {\it target state} $\bm x^*$ at $t=T$ from the {\it source state } $\bm x_0$, i.e., $\bm \xi(T)=\bm x^*$ and $\bm \xi(0)=\bm x_0$, meaning that $\bm x_0\in C(\bm x^*)$. Additionally, we assume that if $\bm \xi(t)$ intersects with any $\mathcal M_{i}$, the intersection is transversal \footnote{In the followings, we assume the transversality of the intersections between $\bm \xi(t)$ and $\mathcal M_i$.}. As $\bm \xi (t)$ is the solution of Eq.~\eqref{eq:std}, $\bm \xi (t)$ is given by 
\begin{eqnarray}
    \bm \xi(t)&=&\bm \xi(0) + \mathbb S \int_0^t\bm u(s)\odot \bm f(\bm \xi(s))\odot \bm p(\bm \xi(s))ds. \label{eq:xi1}
\end{eqnarray}

Next, we fragment the path. Note that the union of all the closures of the direction subsets, $\overline{W}_{\bm \sigma}$, covers $\mathbb R_{> 0}^N$, and $\bm \xi(t)$ intersects transversally to the balance manifolds, if any. Thus, we can divide the interval $[0,T]$ into $L+1$ segments $I_i=(\tau_i,\tau_{i+1}), (0\leq i\leq L)$ so that $\bm \xi(t)\in W_{\bm \sigma^{(i)}}$ holds for $t\in I_i$. In each interval $I_i$, Eq,~\eqref{eq:xi1} is simplified as 
\begin{eqnarray}
    \bm \xi(t)&=&\bm \xi(\tau_i)+\label{eq:xi2}\\&&\mathbb S \bm \sigma^{(i)}\odot \int_{\tau_i}^{t}\bm u(s)\odot \bm f(\bm \xi(s))\odot|\bm p(\bm \xi(s))|ds\nonumber
\end{eqnarray}

Recall that all the functions inside the integral are non-negative. Thus, $\tilde{\bm u}(t):=\bm{u}(t)\odot\bm f(\bm \xi(t)) \odot|\bm p(\bm \xi(t))|$ satisfies $\tilde{\bm u}(t)\in \mathbb R_{\geq 0}^R$. This means that we can consider the ramped function $\tilde{\bm u}(t)$ as a control parameter. By introducing $\bm U_i(t):= \int_{\tau_i}^t \tilde{\bm{u}}(s)ds$, Eq,~\eqref{eq:xi2} is further simplified as 
\begin{eqnarray}
    \bm \xi(t)=\bm \xi(\tau_i)+\mathbb S \bm \sigma^{(i)}\odot \bm U_i(t), \ t \in I_i\label{eq:xi3}
\end{eqnarray}

Therefore, for the solution of Eq.\eqref{eq:std}, $\bm \xi(t)$ with $\bm \xi(0)=\bm x_0$ and $\bm \xi(T)=\bm x^*$, the following holds; There exist a set of time-intervals $\{I_i\}_{i=0}^{L}$, reaction directionalities $\{\bm \sigma^{(i)}\}_{i=0}^L$, and non-negative-valued functions $\{\bm U_i(t)\}_{i=0}^{L}$ such that path in the interval $t\in I_i$ is represented in the form of Eq,~\eqref{eq:xi3}.

The converse of the above argument holds. Suppose that there is a continuous path $\bm\xi\subset\mathbb R_{> 0}^N$ with $\bm x_0$ and $\bm x^*$ as its endpoints. Since $\bm \xi$ is a continuous path, it is continuously parameterizable by $t\in[0,T]$ so that $\bm \xi(0)=\bm x_0$ and $\bm \xi(T)=\bm x^*$ hold. By following the above argument conversely, we can see that if there exist $\{I_i\}_{i=0}^{L}$, $\{\bm \sigma^{(i)}\}_{i=0}^L$, and $\{\bm U_i(t)\geq \bm 0\}_{i=0}^{L}$ satisfying Eq,~\eqref{eq:xi3} and ${\rm d}\bm U_i(t)/{\rm d}t\geq \bm 0$,
%\begin{eqnarray}
%\bm \xi (t)=\bm \xi (\tau_i)+\mathbb S \bm \sigma^{(i)}\odot \bm U_i(t),\  t\in I_i
%\end{eqnarray}
there exists a control $\bm u(t)$ given by
\begin{equation}
\bm u(t)=\frac{{\rm d}\bm U_i(t)/{\rm d}t}{\bm f(\bm x(t))\odot|\bm p(\bm x(t))|},\ t\in I_i
\end{equation}
and $\bm u(\tau_i)=\lim_{t\to\tau_i-0}\bm u(t)$, where the division is performed element-wise. With this $\bm u(t)$, $\bm \xi (t)$ satisfies Eq,~\eqref{eq:std}. Thus, $\bm x_0$ is an element of the controllable set $C(\bm x^*)$.

The summary of the above conditions leads to the definition of the {\it single stoichiometric path}.
%This serves a simple method to compute the controllable set. Suppose that there are the {\it source state} $\bm x_0$, the {\it target state } $\bm x^*$, and the {\it direction sequence} $\{\bm \sigma^{(i)}\}_{i=0}^L$ with $\sgn \bm p(\bm x_0)=\bm \sigma^{(0)}$, $\sgn \bm p(\bm x^*)=\bm \sigma^{(L)}$, and $\bm \sigma^{(i)}\neq \bm \sigma^{(i+1)}$. We now introduce {\it the single stoichiometric path} from $\bm x_0$ to $\bm x^*$ with signs $\{\bm \sigma^{(i)}\}_{i=0}^L$.
\begin{defi}[Single Stoichiometric Path]\label{def:path}
A continuously parameterized path $\bm \xi(t)\subset \mathbb R^N_{>0}$ is called a single stoichiometric path from $\bm x_0$ to $\bm x^*$ with signs $\{\bm \sigma^{(i)}\}_{i=0}^L$ if there exists

\begin{enumerate}[a.]
\item $\{I_i=(\tau_i,\tau_{i+1})\}_{i=0}^L, (\tau_0=0,\tau_{L+1}=T,\tau_i<\tau_{i+1})$,
\item $\{\bm \sigma^{(i)}\}_{i=0}^L$, \\$(\sgn \bm p(\bm x_0)=\bm \sigma^{(0)}$, $\sgn \bm p(\bm x^*)=\bm \sigma^{(L)}, \bm \sigma^{(i)}\neq \bm \sigma^{(i+1)})$,
\item  $\{\bm U_i(t)\geq \bm 0\}_{i=0}^L$ 
\end{enumerate}
satisfying the following conditions;
\begin{enumerate}[1.]
    \renewcommand{\labelenumiii}{\arabic{enumiii}.}
\item $\bm \xi(0)=\bm x_0$ and $\bm \xi(T)=\bm x^*$.
\item $\bm \xi(t)\in W_{\bm \sigma^{(i)}}$ for $t\in I_i$.
\item $\bm \xi(t)=\bm \xi(\tau_i)+\mathbb S \bm \sigma^{(i)}\odot \bm U_i(t)$ for $t\in I_i$.
\item ${\rm d}\bm U_i(t)/{\rm d}t\geq \bm 0$ for $t\in I_i$.
\end{enumerate}
\end{defi}

Let ${\rm SP}_L(\bm x^*)$ denotes the source points of all the single stoichiometric paths to $\bm x^*$ with all possible choices of the sign sequence with a length less than or equal to $L$, and ${\rm SP}(\bm x^*):=\lim_{L\to \infty}{\rm SP}_L(\bm x^*)$. We call ${\rm SP}(\bm x^*)$ {\it the stoichiometric paths} while ${\rm SP}_L(\bm x^*)$ is termed {\it the finite stoichiometric paths} with length $L$. Note that the stoichiometric paths equals to the controllable set.

The stoichiometric paths is a useful equivalent of the controllable set, whereas the conditions 2 and 4 in the definition are laborious to check whether a given path and $\bm U_i(t)$ satisfies them. Thus, we introduce {\it the single stoichiometric ray}. It is easy to compute, and the collection of them gives exactly the stoichiometric paths if the thermodynamic part $\bm p(\bm x)$ of a model equation is linear.

\tcr{Let us show that if $\bm p(\bm x)$ is linear, the single stoichiometric ray guarantees the existence of the single stoichiometric path. Since $\bm p(\bm x)$ is linear, all $W_{\bm \sigma}$'s are either convex polytopes or polyhedra. Suppose there is a set of points}
\begin{equation}
\tcr{\{\bm \xi_0=\bm x_0,\bm \xi_1,\bm \xi_2,\ldots,\bm \xi_{L},\bm \xi_{L+1}=\bm x^*\}}
\end{equation}
\tcr{satisfying} 
\begin{eqnarray}
    &&\tcr{\bm \xi_i\in \overline{W}_{\bm \sigma^{(i-1)}}\cap\overline{W}_{\bm \sigma^{(i)}},}\\
    &&\tcr{\bm \xi_{i+1}=\bm \xi_{i}+\mathbb S \bm \sigma^{(i)}\odot \bm U_i,}\\
    &&\tcr{\bm U_i\geq \bm 0.}
\end{eqnarray}
\tcr{Then, we can construct a line graph $\bm \xi(t)$ parameterized by $t\in[0,T]$ which connects the points $\{\bm \xi_i\}_{i=0}^{L+1}$ as}
\begin{equation}
\tcr{\bm \xi(t):= \bm \xi_{\underline{t}}+\mathbb S \bm \sigma^{({\underline{t}})}\odot \bm U_{{\underline{t}}}s(t),}\label{eq:ray_construct}
\end{equation}
\tcr{where $$s_L(t)=(Lt/T-\lfloor Lt/T\rfloor).$$$\underline{t}$ is defined as $\lfloor t/(LT)\rfloor$ with $\lfloor\cdot \rfloor $ as the floor function. Note that the function $s_L(t)$ is a repetition of the of linear function from $0$ to $1$ for $L$ times (see Fig.~\ref{fig:ray_description} for graphical representations). Here, $\bm  U_{{\underline{t}}}s_L(t)\geq \bm 0$ and ${\rm d}(\bm U_{{\underline{t}}}s_L(t))/{\rm d}t\geq \bm 0$ follow.} Additionally, the points between $\bm \xi_i$ and $\bm \xi_{i+1}$ are in the direction subset $W_{\bm \sigma^{(i)}}$. Thus, $\bm \xi(t)$ and $\bm U_i$ satisfy the conditions 2 to 4 in Def.~\ref{def:path} in $t\in I_i=(iT/L,(i+1)T/L,)$. \tcr{Also, since we choose $\{\bm \xi_i\}_{i=0}^{L+1}$ so that $\bm \xi_0=\bm x_0$ and $\bm \xi_{L+1}=\bm x^*$, the line graph represented by Eq.~\eqref{eq:ray_construct} satisfies all the conditions in Def.~\ref{def:path}. We call this line graph {\it the single stoichiometric ray}. If there is a single stoichiometric ray from $\bm x_0$ to $\bm x^*$, then there is a single stoichiometric path from $\bm x_0$ to $\bm x^*$. Thus $\bm x_0$ is an element of the stoichiometric paths $SP(\bm x^*)$. Note that this is only true if $\bm p(\bm x)$ is linear.}

\begin{figure}[htbp]
    \begin{center}
    \includegraphics[width = 150 mm, angle = 0]{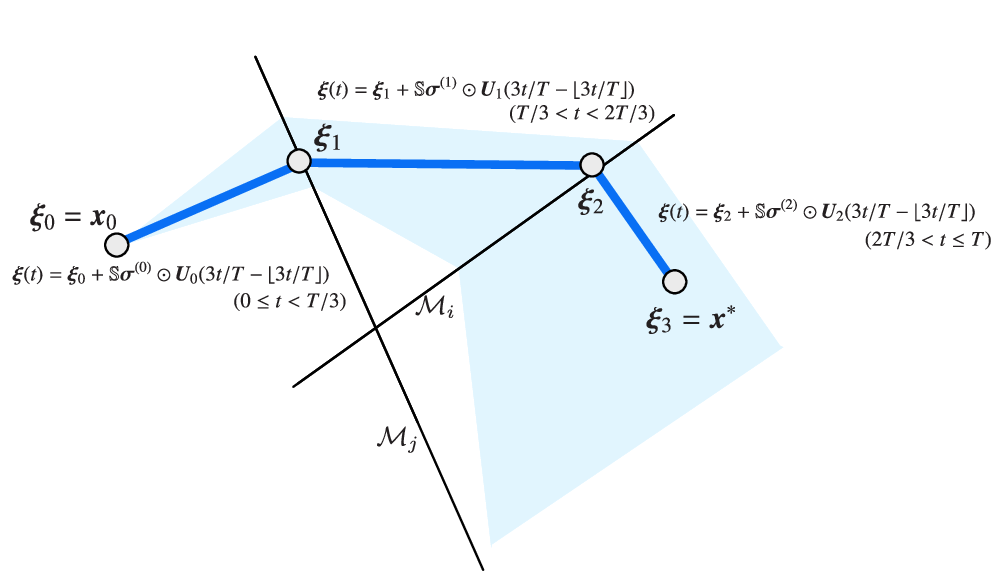}
    \caption{\tcr{A visualisation of a single stoichiometric ray (blue line graph) defined in Eq.~\eqref{eq:ray_construct}. Since the function $s_L(t)=(Lt/T-\lfloor Lt/T\rfloor)$ is a repetition of the linear function from $0$ to $1$ for $L$ times, Eq.~\eqref{eq:ray_construct} represents a line graph connecting the vertices $\{\bm \xi_i\}_{i=0}^{L+1}$. ${\cal M}_i$ and ${\cal M}_j$ are the balance manifolds of reaction $i$ and $j$ respectively. Cyan shaded region is the region reachable from $\bm x_0$ by a single stoichiometric path.}}
            \label{fig:ray_description}
      \end{center}
    \end{figure}

\tcr{Let us formally introduce the single stoichiometric ray.}
\begin{defi}[Single Stoichiometric Ray]\label{def:ray}
    \tcr{A line graph $\bm \xi(t)$ parameterized by $t\in [0,T]$ satisfying the following conditions is called a single stoichiometric ray from $\bm x_0$ to $\bm x^*$} with signs $\{\bm \sigma^{(i)}\}_{i=0}^L$ where $\sgn \bm p(\bm x_0)=\bm \sigma^{(0)}$, $\sgn \bm p(\bm x^*)=\bm \sigma^{(L)}, \bm \sigma^{(i)}\neq \bm \sigma^{(i+1)}$. 

    \begin{enumerate}[1.]
    \item \tcr{$\tau_0=0<\tau_1<\cdots<\tau_{L}<\tau_{L+1}=T,$}
    \item $\bm \xi(0)=\bm x_0$ and $\bm \xi(T)=\bm x^*$,
    \item \tcr{$\bm \xi(\tau_i)\in \overline{W}_{\bm \sigma^{(i-1)}}\cap\overline{W}_{\bm \sigma^{(i)}} \ (1\leq i<L),$ }
    \item \tcr{$\bm \xi(\tau_i+s(\tau_{i+1}-\tau_i))=\bm \xi(\tau_i)+\mathbb S \bm \sigma^{(i)}\odot \bm U_i s, \ (s\in[0,1),\bm U_i\geq \bm 0)$.}
    \end{enumerate}
    \end{defi}

Let us define ${\rm SR}_L(\bm x^*)$ and ${\rm SR}(\bm x^*)$, termed {\it the finite stoichiometric rays with length $L$} and the {\it stoichiometric rays} \footnote{The term ``stoichiometric rays" is derived from the stoichiometric cone. The stoichiometric rays is a generalized concept of the stoichiometric cone \cite{Feinberg2019-hp} defined by
    $${\rm SC}(\bm x)=\{\bm y\in \mathbb R_{\geq0}^N\mid \bm x+\mathbb S \bm v=\bm y,\bm v\geq 0\}.$$
Indeed, if all the reactions of the model are irreversible, the stoichiometric rays of $\bm x^*$ is equivalent to the stoichiometric cone with the replacement of $\mathbb S\bm v$ by $-\mathbb S\bm v$.} in the same manner as the stoichiometric path. Note that only the evaluations of $L$ discrete points are required to determine the existence of a single stoichiometric ray.

Let us remark that ${\rm SR}_L(\bm x^*)={\rm SP}_L(\bm x^*)$ holds if the thermodynamic part $\bm p(\bm x)$ is linear, and thus, the stoichiometric rays is equivalent to the controllable set. If the thermodynamic part is nonlinear, ${\rm SP}_L(\bm x^*)\subseteq{\rm SR}_L(\bm x^*)$ holds. \tcr{If a single stoichiometric path connecting $\bm x_0$ to $\bm x^*$ exists, corresponding single stoichiometric ray exists which is the line graph connecting from 
$\bm x_0$ to $\bm x^*$ via the intersection points between the stoichiometric path and the balance manifolds (see Fig.~\ref{fig:ray_and_path}A). However, the converse does not hold. Note that the stoichiometric path checks the consistency of the reaction direction at all points on the path, while the stoichiometric ray checks the consistency only at the discrete points on the manifolds (Fig.~\ref{fig:ray_and_path}B). Thus, the stoichiometric rays may overestimate the controllable set, while the computational cost is much lower than the stoichiometric paths.}

%This is because the existence of control parameters in a given direction subset, $\bm U_i(t)$, is evaluated based only on points on the balance manifolds. As shown in the illustrative example in Fig.~\ref{fig:overestimate}A, the point $\bm x$ on $\mathcal M_i$ is judged to be reachable to $\mathcal M_j$ even though all possible paths are blocked by another direction subset (the region surrounded by $\mathcal M_k$ in the figure). However, the stoichiometric rays correctly judge the reachability in cases illustrated in Fig.~\ref{fig:overestimate}B.

Let us now show how the stoichiometric rays is computed using the reversible Brusselator model as an illustrative example. The model equation is given by
\begin{eqnarray*}
    \dv{}{t}\begin{pmatrix}
    a\\
    b
    \end{pmatrix}
    &=&\mathbb S\bigl(\bm u\odot \bm f(a,b)\odot \bm p(a,b)\bigr)\\ 
    &=&
    \begin{pmatrix}
    1&-1&1\\
    0&1&-1
    \end{pmatrix}
    \begin{pmatrix}
    u_0\\
    u_1\\
    u_2
    \end{pmatrix}
    \odot
    \begin{pmatrix}
    a_0\\
    1\\
    a^2
    \end{pmatrix}
    \odot
    \begin{pmatrix}
    1-k_0 a\\
    a-k_1 b\\
    b-k_2 a
    \end{pmatrix}.
\end{eqnarray*}
Note that the thermodynamic part vector $\bm p$ in the reversible Brusselator is a linear function of $a$ and $b$ whereas the reaction rate function vector $\bm f\odot \bm p$ is a nonlinear function.

The finite stoichiometric rays for the two choices of the target state $\bm x^*$ are shown in Fig.~\ref{fig:brussel} (For the computational procedure, see SI text section S1 and SI Codes). In Fig.~\ref{fig:brussel}, ${\rm SR}_L(\bm x^*)\backslash{\rm SR}_{L-1}(\bm x^*) $ is filled with distinct colors. As shown in Fig~\ref{fig:brussel}A, it is possible to reach the point $\bm x^*=(0.6,0.6)$ from anywhere in the phase space, ${\rm SR}_L(\bm x^*)=\mathbb R_{> 0}^2$ ($L\geq 3$), while we need to choose the starting point from the specific region for the other target point $\bm x^*=(1.5,1.5)$ as in Fig.~\ref{fig:brussel}B because ${\rm SR}_L(\bm x^*)\subsetneq \mathbb R_{> 0}^2$ ($L\geq 0$) holds for the point.

\begin{figure}[htbp]
    \begin{center}
    \includegraphics[width = 150 mm, angle = 0]{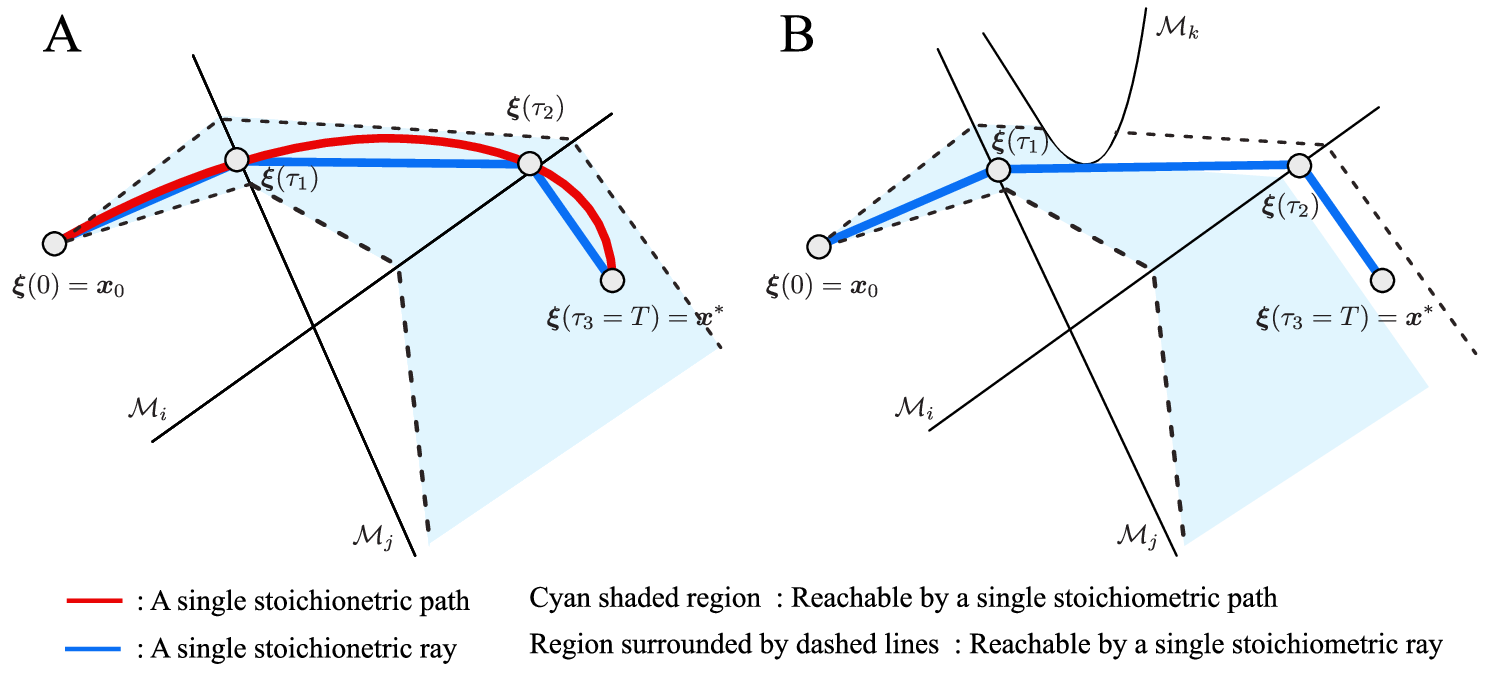}
    \caption{\tcr{(A) A single stoichiometric path (red) and a single stoichiometric ray (blue) from $\bm x_0$ to $\bm x^*$ with $2$ reaction direction flips. The cyan shaded region is the reachable region from $\bm x_0$. The single stoichiometric ray is an approximation of the single stoichiometric path with a line graph. ${\cal M}_i$ and ${\cal M}_j$ are the balance manifolds of reaction $i$ and $j$, respectively. (B) A case where there is no single stoichiometric path from $\bm x_0$ to $\bm x^*$, but there is a single stoichiometric ray. The cyan shaded region is reduced by the blocking by the balance manifold ${\cal M}_k$ (here we assume that the transition through the interior of ${\cal M}_k$ is not possible), and thus the stoichiometric paths from $\bm x_0$ to $\bm x^*$ do not exist. However, since the stoichiometric rays evaluate the existence of the states at the boundary of the direction subsets, ignoring whether the paths are blocked inside the direction subset, a single stoichiometric ray is judged to exist. The region surrounded by the dashed lines is the reachable region by a single stoichiometric ray. Note that not only the single stoichiometric path shown in red in panel A is blocked, but also all paths from $\bm x_0$ to $\bm x^*$.}}
        \label{fig:ray_and_path}
      \end{center}
    \end{figure}

\begin{figure}[htbp]
\begin{center}
\includegraphics[width = 120 mm, angle = 0]{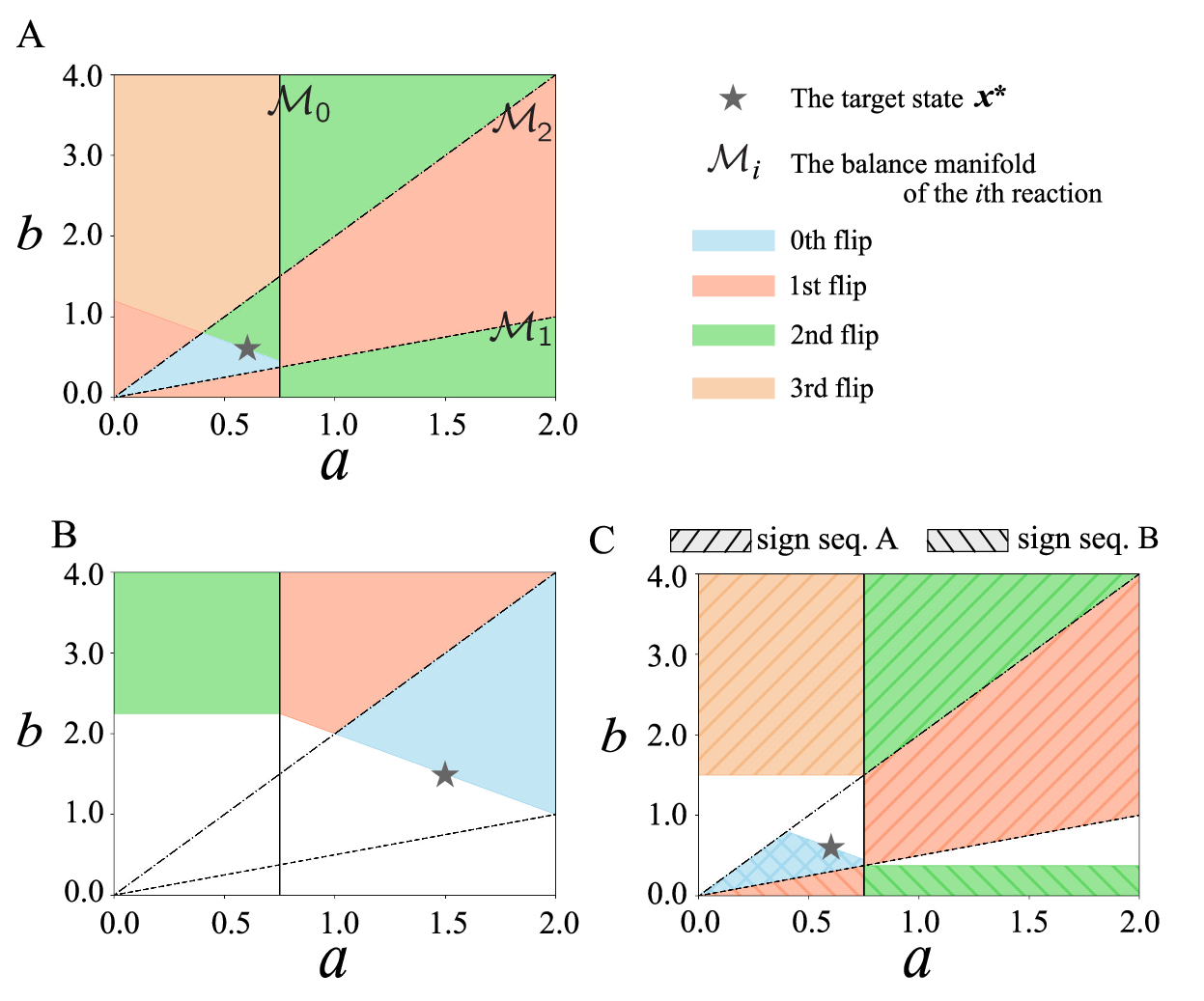}
\caption{The stoichiometric rays of the reversible Brusselator model. The balance manifold $\mathcal M_0,\mathcal M_1,$ and $\mathcal M_2$ are represented by the solid, dash, and dot-dash lines, respectively. The phase space is restricted to $[0,2]\times [0,4]$ for computation. (A) The stoichiometric rays of $\bm x^*=(0.6,0.6)$. The regions covered for the first time at each flip of the reaction direction are filled with different colors. The whole region is covered within $3$ flips. (B) The stoichiometric rays of $\bm x^*=(1.5,1.5)$. In this case, the bottom left region is never covered even if we allow arbitrarily number of flips. (C) The unions of all single stoichiometric rays of $\bm x^*=(0.6,0.6)$ with two example sign sequences are depicted. The sign sequence A and B follow are $(1,-1,-1)\to(-1,-1,-1)\to(-1,-1,1)\to(1,-1,1)$ and $(1,-1,-1)\to(1,1,-1)\to(-1,1,-1)$, respectively. 
%(D) The total volume of the union of the rays is plotted against the flip counts. 
$\bm k=(4/3,2,2)$ is used.}
    \label{fig:brussel}
  \end{center}
\end{figure}

\section{Death of a toy metabolic model}
\subsection{Computing the dead set}

Next, we apply the stoichiometric rays to a toy model of metabolism and compute its dead set. The model is an abstraction of the glycolytic pathway consisting of the metabolites X, Y, ATP, and ADP (a network schematic and reaction list are shown in Fig.~\ref{fig:glycolysis}A).

The rate equation is given by \footnote{For a derivation of the linear thermodynamic part for the reaction ${\rm R}_4$, see SI text section S2.}
\begin{eqnarray}
    &&\dv{}{t}\begin{pmatrix}
    x\\
    y\\
    z\\
    w
    \end{pmatrix}
    =\mathbb S\bigl(\bm u\odot \bm f(\bm x)\odot \bm p(\bm x)\bigr)\label{eq:metabolic}\\ 
    &=&
    \begin{pmatrix}
    -1&-1&-1&0&0&0\\
    0&0&1&-1&-1&0\\
    1&1&-1&0&3&-1\\
    -1&-1&1&0&-3&1
    \end{pmatrix}
    \begin{pmatrix}
    u_0\\
    u_1\\
    u_2\\
    u_3\\
    u_4\\
    u_5
    \end{pmatrix}
    \odot
    \begin{pmatrix}
    xw-k_0 z\\
    xw\\
    xz-k_2 yw\\
    y\\
    yw-k_4 z\\
    z-k_5w
    \end{pmatrix},\nonumber
\end{eqnarray}
where $x,y,z,$ and $w$ represent the concentrations of metabolites X, Y, ATP, and ADP, respectively. In the second line of the above equation, $\bm f(\bm x)=\bm 1$ is omitted. The thermodynamic part of the model is nonlinear. As the total concentration of ATP and ADP is a conserved quantity in the model, we replace $w$ by $z_{\rm tot}-z$ where $z_{\rm tot}$ is the total concentration of the two (we set $z_{\rm tot}$ to unity in the following).

\begin{figure}[htbp]
    \begin{center}
    \includegraphics[width = 120 mm, angle = 0]{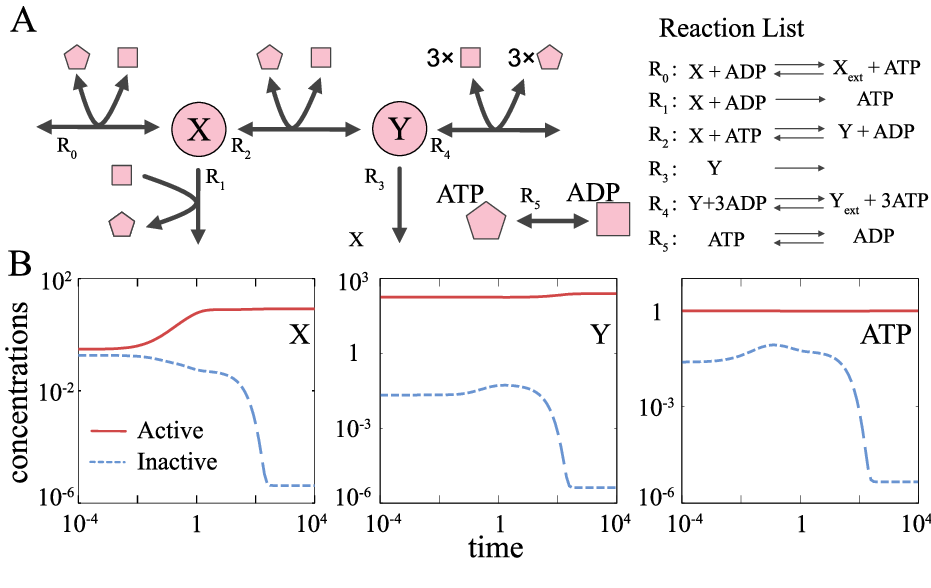}

    \caption{(A) A schematic of the reaction network. The list of the reactions is shown on the right side of the diagram. (B) An example of dynamics of metabolites X, Y, and ATP converging to the active (red), and inactive attractors (blue), respectively. Parameters are set to $\bm u=(1.0,10.0,0.1,0.01,1.0,0.1)$ and $k_0=10.0,k_2=0.1,k_4=1,k_5=10^{-5}$.}
        \label{fig:glycolysis}
      \end{center}
    \end{figure}

As shown in Fig.~\ref{fig:glycolysis}B, the model exhibits bistability with a given constant $\bm u$. On one attractor (the endpoint of the red line), X molecules taken up from the environment are converted into Y and secreted into the external environment resulting in the net production of a single ATP molecule per single X molecule. In contrast at the other attractor (the endpoint of the blue line), almost all X molecules taken up are secreted to the external environment via reaction ${\rm R}_1$. Y is also taken up from the external environment while consuming 3ATPs. Note that reaction ${\rm R}_0$ consumes a single ATP molecule to take up a single X molecule, and the reaction ${\rm R}_1$ produces a single ATP molecule. For the model to obtain a net gain of ATP, X must be converted into Y by consuming one more ATP molecule via reaction ${\rm R}_2$. However, once the model falls into the attractor where ATP is depleted (endpoint of the blue line in Fig.~\ref{fig:glycolysis}B), it cannot invest ATP via ${\rm R}_2$, and the model is ``dead-locked" in the state. Hereafter, we refer to the former and latter attractors as the active attractor $\bm x^A$ and inactive attractor $\bm x^I$, respectively.

Now, we compute the dead states of the model according to Def.~\ref{def:death}. In this framework, we first need to select the representative living states. Here, we suppose that the active attractor $\bm x^A$ is the representative living state. %\footnote{Note that selecting the basin of attraction of $\bm x^A$ as the representative living states $X^*$ is equivalent to setting $X^*=\{\bm x^A\}$ because the basin of attraction is obviously a subset of the controllable set. Further, we can add the states in the basin of attraction of another attractor $\tilde{\bm x}^A$ to $X^*$ without changing the viability result if a continuous change of $\bm u$ in the parameter space leads to the continuous change of the attractor from $\bm x^A$ to $\tilde{\bm x}^A$ in the phase space.\label{fot:att}}.
According to Def.~\ref{def:death}, we compute the complement set of the stoichiometric paths of the active attractor, $D(\bm x^A)=\mathbb R_{> 0}^N \backslash {\rm SP}(\bm x^A)$. This complement set corresponds to the dead set with $\bm x^A$ as the representative living state. Note that choosing $\bm x^A$ as the representative living state is equivalent to choosing mutually reachable attractors from/to $\bm x^A$ by continuously changing the enzyme activities and their basin of attraction (Recall Eq.~\eqref{eq:thm_basin}). By selecting a single-point attractor as the representative living state, we can effectively select a large region in the phase space.

Due to the challenges associated with computing the stoichiometric paths, we utilize the stoichiometric rays for the computation. Let us term $\tilde{D}(\bm x^A)=\mathbb R_{> 0}^N \backslash {\rm SR}(\bm x^A)$ the {\it underestimated dead set}. As the stoichiometric rays is a superset of the stoichiometric paths, the underestimated dead set is a subset of the dead set. Thus the stoichiometric rays does not show a false negative of returnability; If $\bm x\in \tilde{D}(\bm x^A)$ holds, then $\bm x$ is guaranteed to be dead with respect to the active attractor. 

\begin{figure}[htbp]
\begin{center}
\includegraphics[width = 120 mm, angle = 0]{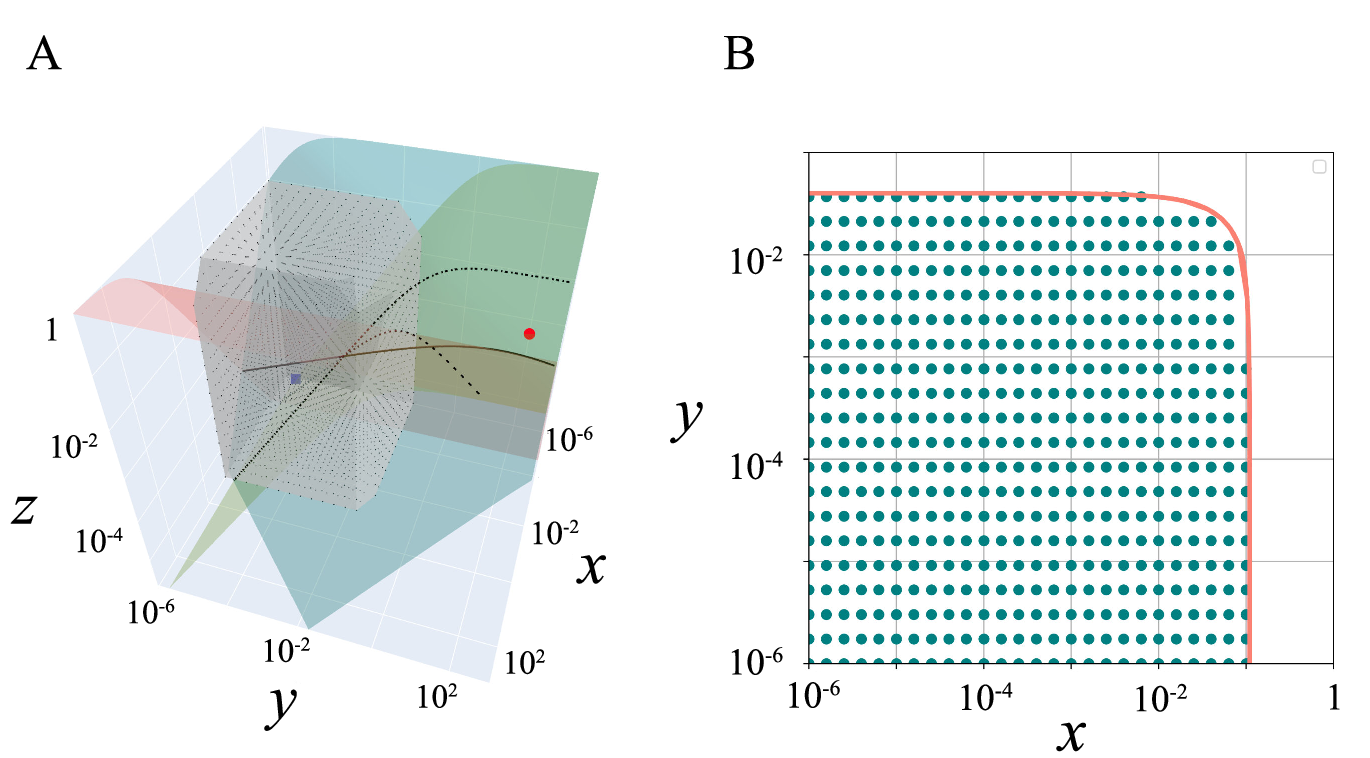}
\caption{(A) The active attractor (red point), the inactive attractor (blue box), and the points in the underestimated dead set (black dots) are plotted. The gray rectangle is the convex hull of the dead states. The balance manifolds $\mathcal M_0$, $\mathcal M_2$, and $\mathcal M_4$ are drawn as the salmon, green-blue, and olive-colored surfaces. The balance manifold $\mathcal M_5$ is not illustrated because it is not necessary for the argument. The intersections of the balance manifold pairs $(\mathcal M_0$, $\mathcal M_2)$, $(\mathcal M_2$, $\mathcal M_4)$, and $(\mathcal M_4$, $\mathcal M_0)$ are represented by the solid, dash, and dot-dash curves, respectively. (B) An analytic estimate of the boundary of the underestimated dead set (red curve) and the points in the underestimated dead set (green-blue points) are plotted. $\log_{10}z$ value is fixed to $-1.7$. We set $L=4$. The same $\bm k$ values with Fig.~\ref{fig:glycolysis} are used.}
    \label{fig:nonreturn}
  \end{center}
\end{figure}

In Fig.~\ref{fig:nonreturn}A, the underestimated dead set with a finite $L$, $\tilde{D}_{L}(\bm x^A)=\mathbb R_{> 0}^N \backslash {\rm SR}_L(\bm x^A)$ is depicted in the phase space (For the computational procedure, see SI text section S3 and SI Codes). The points inside the underestimated dead set cannot return to the active attractor regardless of how the enzyme activities are modulated with a given number of possible flips. The inactive attractor is contained in the underestimated dead set, and thus the inactive attractor is indeed the ``death'' attractor at least with the $L=4$ flips \footnote{We computed the underestimated dead set for $L=6$ with a lower resolution of the point sampling due to the computational cost. We confirmed that the increase of $L$ from $4$ to $6$ does not make the underestimated non-returnable set larger. (see Fig.~S3).}. Also, the boundary of the overestimated controllable set $\tilde{C}_{L}(\bm x^A)$ and the underestimated dead set $\tilde{D}_{L}(\bm x^A)$ now gives an estimate of the SANZ surface, $\tilde{\Gamma}(\bm x^A):=\partial \tilde{C}_{L}(\bm x^A)\cap \partial \tilde{D}_{L}(\bm x^A)$ \footnote{\tcr{The non-emptiness of the dead set is due to the fact that negative-valued controls are not allowed. It is shown from the rank relation of the Lie algebra of the model that if negative values are allowed for the control parameters to take, controls from any state to any state are feasible \cite{nijmeijer1990nonlinear,coron2007control}. In particular, the model shows a bifurcation at $\bm u\approx (1.0,10.0,0.1,-0.08,1.0,0.1)$, where the inactive attractor becomes unstable, and the model shows the unistability of the active attractor. Thus, modulating $u_3$ to $\approx -0.08$ is a way to revive it.}}.

In Fig.~\ref{fig:nonreturn}B, we plot the points in the underestimated dead set with an analytically estimated SANZ surface of the set on the 2-dimensional slice of the phase space where $z$ value is set to a single value. The SANZ surface is calculated from the thermodynamic and stoichiometric constraints that each single stoichiometric ray should satisfy. The curves match well with the boundary of the underestimated dead set (for the detailed calculation, see SI Text section S4).

\subsection{The global transitivity of the toy metabolic model}

The application of the stoichiometric rays is not limited only to computations of the dead set with pre-defined representative living states. The stoichiometric rays enables us to investigate the mutual transitivity among states without a priori definition of living states. In this section, we present a transition diagram of the toy metabolic model \footnote{The computation of the transition diagram corresponds to the computation of the equivalent classes and partial order $(\looparrowright,\mathbb R_{>0}^3/{\leftrightharpoons})$ defined in the appendix section.}. The transition diagram thus dictates the global transitivity of the model. This provides insights for choosing the representative living states, and in addition, criteria for the model selection suitable for investigating cell death.

For the computation of the transition diagram (procedure of the computation is illustrated in Fig.~\ref{fig:hasse}A), we sampled several points in the phase space (Fig.~\ref{fig:hasse}B) in addition to the active and inactive attractors, and computed controllability for all pairs of the sampled points. Note that we computed only the controllability to the active attractor in the previous part, but here we compute the controllability of all pairs of points mutually. With the computation, we construct the directed graph $G_L$ where the nodes are the sampled points and the directed edges represent the controllability from the source point to the target point, i.e, we put the directed edge from $\bm x$ to $\bm y$ if $\bm x\in {\rm SR}_L(\bm y)$ holds.

Next, we computed the condensation of $G_L$ which, denoted $\tilde{G}_L$. Condensation is a directed graph in which strongly connected components (SCCs) of $G_L$ are contracted to a single node. The nodes of $G_L$ within the same SCC are mutually reachable \footnote{\tcr{Each node in the condensation graph $\tilde{G}_L$ corresponds to an equivalent class (see Appendix).}}. Hereafter, we refer to the condensation graph $\tilde{G}_L$ as a transition diagram. 

The transition diagram is shown in Fig.~\ref{fig:hasse}C \footnote{\tcr{The transition diagram describes the partial order relation $\looparrowright$ on $\mathbb R_{>0}^3/{\leftrightharpoons}$ (see Appendix).}}. In addition, because the condensation graph is a directed acyclic graph, the diagram highlights the hierarchy of the state transitions or the ``potential'' of the states. The active and inactive attractors are in the SCCs represented by the yellow- and pink squares in Fig.~\ref{fig:hasse}C, respectively. Let us refer to the pink SCC as the terminal SCC, $C^*$, since it has only in-edges, and thus, the states in $C^*=\Omega(\mathbb R_{>0}^3)$ (cf.~Def.~\ref{def:terminal} in Appendix) cannot escape from the set regardless of the control $\bm u(t)$. %\tcr{The possible transitions in the phase space by the modulations of the control parameters with a given number of flips is fully captured by the transition diagram. The phase space has uncountably infinite number of states. However, the transition diagram has the countably infinite, or even finite number of states which dictates the global transitivity of the model (see Appendix for more details).}% While we carried out the computations of the transitivity for a limited number of points in the phase space due to the computational cost.}

The computation of the transition diagram provides a guide for selecting the representative living states, and in addition, a criterion for the model selection suitable for investigating cell death. It is commonly believed that controlling to the living states is difficult and possible only from limited states. Also, all living and non-living systems can be non-living (dead state), i.e., control to the dead state should always be possible regardless of the source state. If we accept these beliefs, the representative living state should be taken from the SCCs except for the terminal SCC. When we additionally require that the representative living state be the attractor (i.e., the state is stable without temporal modulation of enzymatic activities), the representative living state should be taken from the SCC containing the active attractor (the yellow square in Fig.~\ref{fig:hasse}C). 

In the toy metabolic model, the active attractor and the inactive attractors are in the non-terminal SCC and the terminal SCC, respectively. However, this is not always the case for any biochemical model. In order to study cell death employing mathematical models, we may require the models to satisfy conditions such as the one above. The transition diagram is useful method for checking if a given model satisfies requirements as the diagram captures the global transitivity in the model.

\begin{figure}[htbp]
    \begin{center}
    \includegraphics[width = 150 mm, angle = 0]{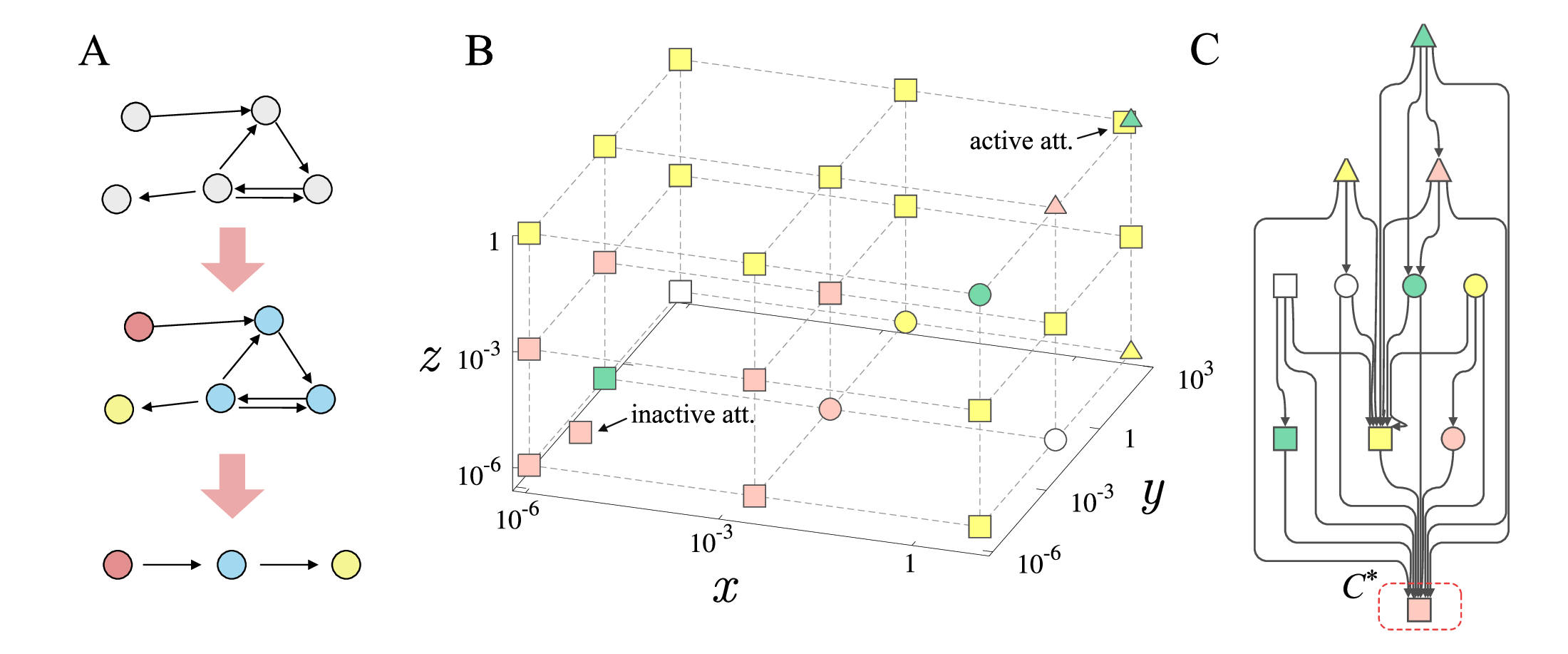}

\caption{\tcr{The transition diagram of the discrete states of the toy metabolic model. (A) Construction of the transition diagram. First, the controllability from/to all pairs of points are computed and the directed graph is drawn (top graph). In the graph, the directed edge represents that there are controls from the state at the tail of the arrow to the state at the head of the arrow. Then the nodes are grouped into strongly connected components (middle graph). The component graph is obtained by merging the nodes belonging to the same strongly connected components into a single node (bottom graph).} (B) The points for which the mutual transitivity is computed. The points are grouped and labeled by the strongly connected components (SCCs). The dashed lines are eye guides. The active attractor (yellow square) and the inactive attractor (pink square) are shown together. Note that the green triangle is on the grid, but not the active attractor. (C) The transition diagram of the SCCs. Only the two SCCs (pink- and yellow squares) contain more than one state. The terminal SCC (pink square, indicated by the red-dashed rectangle), $C^*$, has only in-edges. The same $\bm k$ values with Fig.~\ref{fig:glycolysis} are used, and $L=4$.}
        \label{fig:hasse}
      \end{center}
    \end{figure}

\section{Discussion}
Questions such as ``what is life?'' and ``what is death?'' relate to how we can characterize living and dead states. Providing a straightforward answer for such a characterization is challenging even for mathematical models of cells. However, quantitative evaluation of the ``life-death boundary'' (SANZ hypersurface), or in other words, characterization of the difference between life and death can be possible. In the present manuscript, we proposed a theoretical framework to study the dead set and SANZ hypersurface with the development of a tool for computing it.

In the present framework, the dead state is defined based on the controllability to the representative living state. Thanks to the useful mathematical features of controllability relationship, $\bm x\rightharpoonup \bm y$, the death judgment is robust on the choice of the representative living state; The choice of representative living state is arbitrary, as long as they are chosen from the same equivalence class. Further, Thm.~\ref{thm:termial} (see Appendix) states that even if there are uncountably many living states, the representative living states can be countably many, or even, a finite number. 

Our strategy for defining the dead state is to use the law of exclusion, of the form ``what is not living is dead''. This allows us to avoid direct characterization of dead states. In this sense, our characterization of death is similar to that of the regrow experiment, rather than the staining assay in which the dead state is characterized by stainability. For such indirect characterizations, it is necessary to determine whether a cell will eventually regrow. This is an extremely hard challenge in experimental systems. However, by focusing only on mathematical models, we can determine, in principle, whether the cell will regrow. The stoichiometric rays provides a reasonable method for computing this.

\tcr{Note that the present framework says nothing about which states should be chosen as representative living states. We assume that information about which states should be considered ``living'' is given by other sources, such as experiments and observations. The present framework is intended to provide a quantitative criterion for the irreversible transition from the living states to the dead states with respect to a given representative living states. Yet, the living/dead judgment is not that sensitive to the choice of representative living states. We have shown that choosing several states as the representative living states is typically equivalent to choosing a large variety of states as the representative living states (see Thm.~\ref{thm:termial} in the Appendix). Also, if the dead set with respect to the first choice of representative living states $X$ has states that are ``obviously living'' by other criteria, one can update the representative living states to $X'$ by adding those states to $X$ and compute the dead set $D(X')$. The present framework does not define life and death as if it were a transcendent being, but rather computes the dead states and the SANZ hypersurface based on one's working hypothesis of ``what is living''.}

The heart of the stoichiometric rays is that the modulation of enzymatic activity cannot directly change the directionalities of the reactions. As a consequence of this feature, the balance manifolds and direction subsets are invariant to enzymatic activities. This allows us to efficiently compute the stoichiometric rays. The usefulness of the stoichiometric rays was demonstrated by using two simple nonlinear models of chemical reactions. Especially, in the toy metabolic model, we showed that the stoichiometric rays is a powerful tool for computing the dead set in which no control can return the state to the active attractor with a given number of flips. 

%Determining the returnability of the state is crucial in understanding biological irreversibility, especially cell death. There are several criteria for the death of microbial cells. A straightforward criterion is checking the growth rate at a given time point. However, the doubling times of several microbes in the deep sea are estimated up to a million years \cite{Hoehler2013-pi}. Viability check of dead cells by staining such as propidium iodide is a common method, but it is not always reliable \cite{Umetani2021-vz} especially if the lipid membrane is intact. The regrow experiment is a direct check of the viability, but it is repeatedly reported that cells can show a long lag time before the regrowth \cite{
%Levin-Reisman2010-dz,Kaplan2021-mu}, and there is no guiding principle for how long we should wait for the regrowth. The stoichiometric rays can provide a quantitative criterion for the death of {\it in silico} cells. Although it is not directly applicable to the experimental systems, we can quantify the ``life-death boundary'', or rather, we like to borrow a word from Buddhism and call it {\it SANZ Hypersurface} \footnote{The SANZ River is a mythical river in the Japanese Buddhist tradition that represents the boundary between the world of the living and the afterlife.}, of {\it in silico} cells using the stoichiometric rays, and it may contribute to the development of a new criterion for cell death based on the intracellular states and mathematical theory of ``death''.

Thermodynamic (ir)reversibility is believed to be key to grasping the biological (ir)reversibility such as the irreversibility of some diseases and death. However, there is no system-level link between thermodynamic and biological reversibilities except in a few cases where simple causality leads to biological irreversibility such as in transmissible spongiform encephalopathy \cite{collins2004transmissible}. The stoichiometric rays enables us to bridge the (ir)reversibility of each reaction to the system-level reversibility; Recall that the SANZ surface of the dead set is calculated from the thermodynamic constraint, i.e., directionality of the reactions, and the stoichiometric constraint. Further studies should pave the way to unveil the relationship between thermodynamic irreversibility and biological irreversibility. This may reveal the energy (or ``negentropy" \cite{schrodinger1944life}) required to maintain cellular states away from the SANZ hypersurface \cite{pirt1965maintenance}.

\tcr{We do not intend to extend the scope of the framework indiscriminately, while it is noteworthy that the framework does not assume any specific mechanism of cell death. There are multiple cell death mechanisms such as apoptosis, necrosis, ferroptosis, pyroptosis, autophagy, erebosis, and so on \cite{gouy2016results,ciesielski2022erebosis,D-Arcy2019-an,Newton2024-hh,Yuan2024-nv,Park2023-zw}. The molecular details and cellular behaviors during death processes differ depending on the mechanism. However, the irreversible transition to a state where multiple cellular functions are lost is common to all mechanisms. Our framework provides a quantitative criterion for the irreversible transition, and thus its applicability is not limited to specific mechanisms.}

Finally, we would like to mention an extension of the present framework to the probabilistic formulation and its implications. In the present deterministic version, the transition from the living state to the dead state is irreversible. In stochastic processes, however, the transition from the state considered dead in the deterministic model to the living states is possible with some probability. Thus, in probabilistic formulation, life and death are not binary categories, but are continuously linked states. Given the small size of cells, the probabilistic formulation is more suitable for describing cellular state transitions. Death is commonly believed to be irreversible. However, if the theory of death reveals the probabilistic nature of life-death transitions, it may have multiple implications for reforming the concept of ``death''.

\appendix*

\section{The choice of the representative living states}\label{sec:theorems}

\tcr{In this section we make some statements about the choice of representative living states. The main claim is that controllability leads to the equivalence relation and the partial order. The equivalence relation shows that there are pairs of states such that choosing one of them as the representative living state results in the same dead set as choosing both of them. The partial order shows that there is a reduction of the representative living state without changing the dead set.} 

Let us start introducing the the controllability relationship.
\begin{defi}[$\rightharpoonup$ and $\rightleftharpoons$]\label{defi:harpoon}
We denote $\bm x \rightharpoonup \bm y$ if $\bm x,\bm y\in \mathbb R_{\geq 0}^N$ satisfy $\bm x\in C(\bm y)$. If $\bm x\rightharpoonup \bm y$ and $\bm y\rightharpoonup \bm x$ holds, we denote $\bm x\rightleftharpoons \bm y$.
\end{defi}  
\begin{prop}\label{prop:preord}
$\rightharpoonup$ is preorder and $\rightleftharpoons$ is equivalence relation on $\mathbb R_{\geq 0}^N$.
\end{prop}
\noindent The proof is straightforward from the definition of the controllable set.

From the proposition, we obtain the following corollary. 
\begin{col}\label{col:preord}
For any $\bm x,\bm y\in \mathbb R_{\geq 0}^N$ satisfying $\bm x\rightharpoonup \bm y$, 
\begin{equation}
    C(\bm x)\subset C(\bm y)
\end{equation}
holds. If $\bm x\rightleftharpoons \bm y$ holds, we have 
\begin{equation}
    C(\bm x)=C(\bm y).
\end{equation} 
\end{col}
\begin{proof}
For any $\bm z\in C(\bm x)$, $\bm z\rightharpoonup\bm x$ holds. From the assumption, we also have $\bm x\rightharpoonup\bm y$. Thus, the transitivity of $\rightharpoonup$ results in $\bm z\rightharpoonup\bm y$. Therefore, $C(\bm x)\subset C(\bm y)$ holds. By repeating the same argument for $\bm y\rightharpoonup \bm x$, we have the second statement. 
\end{proof}

Thereby, if the states $\bm x,\bm y\in \mathbb R_{\geq 0}^N$ are mutually controllable to each other $\bm x \rightleftharpoons \bm y$, then the dead set is the same whether $\bm x$ or $\bm y$ is chosen as the representative living state. This implies that we can use the quotient set $\mathbb R_{\geq 0}^N/{\rightleftharpoons}$ for studying the controllable set of the states instead of the original space $\mathbb R_{\geq 0}^N$. 

\begin{defi}[Quotient set and Equivalence class]\label{defi:equivalence}
    The quatient set of $\mathbb R_{\geq 0}^N$ by $\rightleftharpoons$ is represented by $\mathbb R_{\geq 0}^N/{\rightleftharpoons}$, and the equivalence class of $\bm x$ is denoted by $[\bm x]\in \mathbb R_{\geq 0}^N/{\rightleftharpoons}$.
\end{defi}

\begin{defi}[$\looparrowright$]\label{defi:loop}
    For any $[\bm x],[\bm y]\in \mathbb R_{\geq 0}^N/{\rightleftharpoons}$, if $\bm x\in[\bm x]$ and $\bm y\in [\bm y]$ exist such that $\bm x\rightharpoonup \bm y$ holds \tcr{(see Def.~\ref{defi:harpoon})}, we denote $[\bm x]\looparrowright [\bm y]$.
\end{defi}

\begin{prop}\label{prop:partord}
$\looparrowright$ is partial order on $\mathbb R_{\geq 0}^N/{\rightleftharpoons}$.
\end{prop}
\begin{proof} 
    Note that if we have $[\bm x]\looparrowright [\bm y]$, $\bm x \rightharpoonup \bm y$ holds for any $\bm x\in[\bm x]$ and $\bm y\in[\bm y]$.
    
    Reflexibility: Since $\bm x\rightharpoonup\bm x$ holds, $[\bm x]\looparrowright [\bm x]$ holds. 

    Transitivity: If $[\bm x]\looparrowright [\bm y]$ and $[\bm y]\looparrowright [\bm z]$ hold, $\bm x\rightharpoonup \bm y$ and $\bm y\rightharpoonup \bm z$ hold. Thus, $\bm x\rightharpoonup \bm z$ holds from the transitivity of $\rightharpoonup$, and $[\bm x]\looparrowright [\bm z]$ holds.

    Antisymmetry: If $[\bm x]\looparrowright [\bm y]$ and $[\bm y]\looparrowright [\bm x]$ hold, $\bm x\rightharpoonup \bm y$ and $\bm y\rightharpoonup \bm x$ hold. Thus, $\bm x\rightleftharpoons \bm y$ holds, and we have $[\bm x]=[\bm y]$.
    
\end{proof}   

\begin{col}\label{col:partord}
    For any $[\bm x],[\bm y]\in \mathbb R_{\geq 0}^N/{\rightleftharpoons}$, if $[\bm x]\looparrowright [\bm y]$ holds, we have 
    \begin{equation}
        C([\bm x])\subset C([\bm y]).
    \end{equation}
\end{col}
\begin{proof}
The controllable set of the equivalence class $[\bm x]$ is a union of all points in $[\bm x]$, i.e., $C([\bm x])=\cup_{\bm x'\in [\bm x]}C(\bm x')$. From Cor.~\ref{col:preord}, the controllable set of every $\bm x\in [\bm x]$ is identical. Therefore, we can take a representative state $\bm x\in [\bm x]$ and $C([\bm x])=C(\bm x)$ holds. The same argument for $[\bm y]$ leads to $C([\bm y])=C(\bm y),\  (\bm y\in [\bm y])$. By appliying Cor.~\ref{col:preord} to $C(\bm x)$ and $C(\bm y)$, we obtain $C([\bm x])=C(\bm x)\subset C(\bm y)=C([\bm y])$.
\end{proof}

This corollary states that the ``larger'' equivalence class in terms of $\looparrowright$ \footnote{Here, we say that $[\bm x]$ is ``larger'' than $[\bm y]$ if $[\bm y]\looparrowright [\bm x]$ holds.} has the larger controllable set. This motivates us to introduce the {\it terminal class} of the representative living states $X$.

\begin{defi}[Terminal Class]\label{def:terminal}
    For a given set $X\subset \mathbb R_{\geq 0}^N$, we take the set of the equivalence classes $[\bm x]\in \mathbb R_{\geq 0}^N/{\rightleftharpoons}$ with nonempty intersection with $X$ and denote it by ${\mathcal X}\subset \mathbb R_{\geq 0}^N/{\rightleftharpoons}$ i.e., 
    \begin{equation}
        {\mathcal X}:=\bigl \{[\bm x]\in\mathbb R_{\geq 0}^N/{\rightleftharpoons}\mid [\bm x]\cap X\neq \varnothing\bigr \}.
    \end{equation}

    We call the maximal element of ${\mathcal X}$ with respect to $\looparrowright$ the terminal class of $X$. We denote the set of all terminal classes of $X$ by $\Omega(X)$ \tcr{(for $\mathbb R_{\geq 0}^N/{\rightleftharpoons}$ and $\looparrowright$, see Def.~\ref{defi:equivalence} and Def.~\ref{defi:loop}, respectively.). }
\end{defi}
\noindent Note that an equivalence class $[\bm x^*]\in{\mathcal X}$ being the terminal class means that there is no other equivalence class $[\bm x]\in{\mathcal X},\ ([\bm x]\neq[\bm x^*])$ satisfying $[\bm x^*]\looparrowright[\bm x]$ exists. 

$\Omega(X)$ can be an empty, finite, or infinite set. However, if we can assume that every chain in ${\mathcal X}$ (a totally ordered subset of ${\mathcal X}$) has an upper bound with respect to $\looparrowright$, Zorn's lemma guarantees $\Omega(X)$ being nonempty. To satisfy this condition, it is sufficient that the model (Eq.~\eqref{eq:ode}) has a choice of the constant control parameter $\bm u$ with which the model never shows a divergent behavior. %Since here we deal with a non-negative orthant $\mathbb R_{\geq 0}^N$, the decreasing direction for each variable is bounded at $x_i=0$. Thereby, if the model does not show a divergent behavior in at least a single parameter choice, the chain in ${\mathcal X}$ has an upper bound.

The controllable set of each terminal class $[\bm x^*]$ contains that of any other equivalence class satisfying $[\bm x]\looparrowright [\bm x^*]$. This allows us to reduce the representative living states $X$ into a smaller set $X^*$ without changing the dead set. 

For $X\subset \mathbb R_{\geq 0}^N$, suppose that $\Omega(X)$ is nonempty. Let us take a point from each terminal class in $\Omega(X)$ and denote the set of the chosen points by $X^*\subset \mathbb R_{\geq 0}^N$, that is,
\begin{equation}
    X^*=\bigcup_{[\bm x^*]\in\Omega(X)}\{\bm x^*\},
\end{equation}
where the choise of the representative element from $[\bm x^*]$ is arbitrary. We have the following theorem for the dead set with respect to $X$ and $X^*$ \tcr{(For a graphical representation of the theorem, see Fig.~\ref{fig:equivalence}).}
\begin{thm}\label{thm:termial}
The dead set with respect to $X$ is identical to the dead set with respect to $X^*$, i.e., 
\begin{equation}
    D(X)=D(X^*).
\end{equation}
\end{thm}
%For the proof of the theorem, we need the following lemma.
%\begin{lem}\label{lem:termial}
%For $X\subset \mathbb R_{\geq 0}^N$ and 
%$${\mathcal X}:=\bigl \{[\bm x]\in\mathbb R_{\geq 0}^N/{\rightleftharpoons}\mid [\bm x]\cap X\neq \varnothing\bigr \},$$
%we have the following equality,
%$$C(X)=\bigcup_{[\bm x]\in {\mathcal X}}C([\bm x]).$$
%\end{lem}
%\begin{proof}
%    $\supseteq$: Let us take $\bm x\in [\bm x]$ so that $\bm x\in X$. From the definition of the controllable set, $C(\bm x)\subset C(X)$ holds. From Cor.~\ref{col:preord}, the controllable set of a state is identical to the controllable set of the equivalence class that the point belongs to, i.e., $C(\bm x)=C([\bm x])$. Thereby, $C([\bm x])\subseteq C(X)$ holds.

%    $\subseteq$: Since $X\subset\cup_{[\bm x]\in{\mathcal X}}[\bm x]$ holds, we have $C(X)\subseteq\cup_{[\bm x]\in{\mathcal X}}C([\bm x])$
%\end{proof}
\begin{proof}
Recall that $D(X)=D(X^*)$ is equivalent to $C(X)=C(X^*)$. The controllable set of $X$ is defined as
\begin{equation}
    C(X)=\bigcup_{\bm x\in X}C(\bm x).
\end{equation}
As the controllable set of the equivalence class $C([\bm x])$ and the controllable set of a point $\bm x\in [\bm x]$ are identical (Cor.~\ref{col:preord}), we have
\begin{equation}
C(X)=\bigcup_{[\bm x]\in {\mathcal X}}C([\bm x])\label{eq:proof_termial}
\end{equation}
where ${\mathcal X}$ is what defined in Def.~\ref{def:terminal}. 

For each $[\bm x]\in {\mathcal X}$, if there exists $[\bm y]\in {\mathcal X},\ ([\bm y]\neq [\bm x])$ such that $[\bm y]\looparrowright [\bm x]$ holds, we have $C([\bm y])\subset C([\bm x])$ from Cor.~\ref{col:partord}. Therefore, we can drop $C([\bm y])$ from the union in Eq.~\eqref{eq:proof_termial} without violating the equality. By repeating the dropping process, only the terminal classes remain and we obtain the following equality:
\begin{equation}
C(X)=\bigcup_{[\bm x]\in \Omega(X)}C([\bm x]).
\end{equation}

Additionally, since $C([\bm x])=C(\bm x)$ holds for any representative point $\bm x\in [\bm x]$, we have
\begin{equation}
C(X)=\bigcup_{\bm x\in X^*}C(\bm x).
\end{equation}

\end{proof}

The theorem states a useful property for choosing the representative living states to compute a dead set. Suppose that we have the representative living state $X\subset \mathbb R_{\geq 0}^N$. Without the theorem, we must check the controllability of all points in $X$ to judge the dead state. However, with this theorem, we can take discrete points from each terminal class in $\Omega(X)$ and check the controllability of the chosen points to judge the dead state. Note that the set $X^*$ can also be finite, which enables further efficient computations of the dead set.

An example of practical applications of this theorem is as follows. Suppose that the model has a point attractor $\bm a_{\bm u}$ with a fixed control parameter value $\bm u$, and $\bm a_{\bm u}$ is a continuous function of $\bm u$ in $U\subset \mathbb R^P$. In such cases, taking the point attractor with a single parameter choice $\bm u_0\in U$ is equivalent to taking the basin of attraction  $B(\bm a_{\bm u})$ for all $\bm a_{\bm u}$, as the representative living states, i.e., 
\begin{equation}
D(\bm a_{\bm u_0})=D \Bigl(\bigcup_{\bm u\in U}B(\bm a_{\bm u})\Bigr).\label{eq:thm_basin}
\end{equation}
This is because $\bm a_{\bm u_1}\rightleftharpoons \bm a_{\bm u_2}$ holds for any $\bm u_1,\bm u_2\in U$ and also we have $\bm x\rightharpoonup \bm a_{\bm u}$ for any $\bm x\in B(\bm a_{\bm u})$.

\begin{figure}[htbp]
    \begin{center}
    \includegraphics[width = 120 mm, angle = 0]{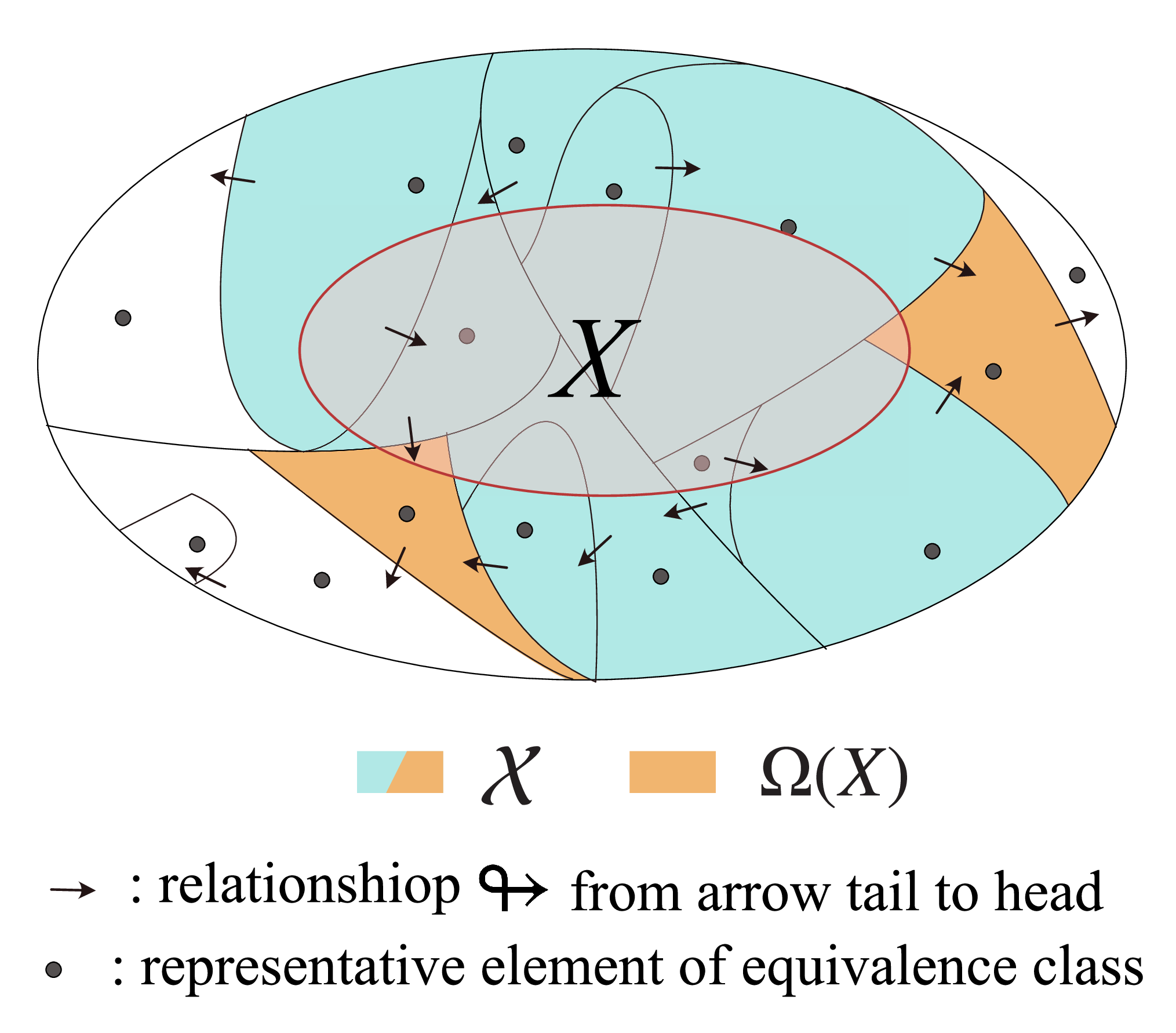}
    \caption{\tcr{The equivalance classes and the representative living states $X$. The large and small ellipses represent $\mathbb R^N_{\geq0}$ and the representative living states $X$, respectively. Each compartment is the equivalence class $[\bm x]\in \mathbb R_{\geq 0}^N/{\leftrightharpoons}$. Where the arrow represents the partial order $\looparrowright$ from the tail of the arrow to the head. Each black point is the representative element of each equivalence class. The set $\cal X$ is the set of equivalence classes filled with cyan or orange. The terminal classes $\Omega(X)\subset \cal X$ is the set of equivalence classes filled with orange. According to Theorem.~\ref{thm:termial}, the controllability to $X$ is fully captured by computing the controllability to the representative elements of $[\bm x]\in\Omega(X)$.}}
        \label{fig:equivalence}
      \end{center}
    \end{figure}

%Overall, the viability judgment by Def.~\ref{def:death} does not alter by the choice of the representative living states as long as the states are chosen from the same attractor branch. Also, if a point in the attractor branch is chosen as a representative living state, we do not need to additonally include any point from the basin of attraction of the attractor in the attractor branch. Thus, selecting a single attactor in the attractor branch is as the representative living states is the minimal setup for the viability judgment with any subset of the attractor branch as the representative living states.

%We need to note that the death defined by Def.~\ref{def:death} depends also on the possible controls. Let $\bm u_i(t)$ denotes the control parameter with the range of value, $U_i\ (i=1,2)$. If $U_1\subseteq U_2$ holds, the dead set with the control $\bm u_1(t)$, $D_1$, is a superset of the dead set with the control $\bm u_2(t)$, $D_2$, i.e., $D_1\supseteq D_2$ because all the control with $\bm u_1(t)$ can be realized by $\bm u_2(t)$, but not the other way around. In the next section, we show that the controllable set is efficiently computable for the enzymatic reaction models.

\section*{Acknowledgments}
We thank Chikara Furusawa and Ikumi Kobayashi for the discussions. This work is supported by JSPS KAKENHI (Grant Numbers JP22K15069, JP22H05403 to Y.H.; 24K00542, 19H05799  to T.J.K.), JST (Grant Numbers JPMJCR2011, JPMJCR1927 to T.J.K.), and GteX Program Japan Grant Number JPMJGX23B4. S.A.H. is financially supported by the JSPS Research Fellowship Grant No. JP21J21415.

\section*{SI Codes}
All codes for the main result are deposited on GitHub (\url{https://github.com/yhimeoka/StoichiometricRay}).


\begin{thebibliography}{0}%
\makeatletter
\providecommand \@ifxundefined [1]{%
 \@ifx{#1\undefined}
}%
\providecommand \@ifnum [1]{%
 \ifnum #1\expandafter \@firstoftwo
 \else \expandafter \@secondoftwo
 \fi
}%
\providecommand \@ifx [1]{%
 \ifx #1\expandafter \@firstoftwo
 \else \expandafter \@secondoftwo
 \fi
}%
\providecommand \natexlab [1]{#1}%
\providecommand \enquote  [1]{``#1''}%
\providecommand \bibnamefont  [1]{#1}%
\providecommand \bibfnamefont [1]{#1}%
\providecommand \citenamefont [1]{#1}%
\providecommand \href@noop [0]{\@secondoftwo}%
\providecommand \href [0]{\begingroup \@sanitize@url \@href}%
\providecommand \@href[1]{\@@startlink{#1}\@@href}%
\providecommand \@@href[1]{\endgroup#1\@@endlink}%
\providecommand \@sanitize@url [0]{\catcode `\\12\catcode `\$12\catcode `\&12\catcode `\#12\catcode `\^12\catcode `\_12\catcode `\%12\relax}%
\providecommand \@@startlink[1]{}%
\providecommand \@@endlink[0]{}%
\providecommand \url  [0]{\begingroup\@sanitize@url \@url }%
\providecommand \@url [1]{\endgroup\@href {#1}{\urlprefix }}%
\providecommand \urlprefix  [0]{URL }%
\providecommand \Eprint [0]{\href }%
\providecommand \doibase [0]{http://dx.doi.org/}%
\providecommand \selectlanguage [0]{\@gobble}%
\providecommand \bibinfo  [0]{\@secondoftwo}%
\providecommand \bibfield  [0]{\@secondoftwo}%
\providecommand \translation [1]{[#1]}%
\providecommand \BibitemOpen [0]{}%
\providecommand \bibitemStop [0]{}%
\providecommand \bibitemNoStop [0]{.\EOS\space}%
\providecommand \EOS [0]{\spacefactor3000\relax}%
\providecommand \BibitemShut  [1]{\csname bibitem#1\endcsname}%
\let\auto@bib@innerbib\@empty
%</preamble>
\end{thebibliography}%


\begin{thebibliography}{10}

    \bibitem{Casolari2018-zo}
    Antonio Casolari.
    \newblock Microbial death.
    \newblock In {\em Physiological models in microbiology}, pages 1--44. CRC press, 2018.
    
    \bibitem{Nakaoka2017-qy}
    Hidenori Nakaoka and Yuichi Wakamoto.
    \newblock Aging, mortality, and the fast growth trade-off of schizosaccharomyces pombe.
    \newblock {\em PLoS Biol.}, 15(6):e2001109, June 2017.
    
    \bibitem{Spivey2017-no}
    Eric~C Spivey, Stephen~K Jones, Jr, James~R Rybarski, Fatema~A Saifuddin, and Ilya~J Finkelstein.
    \newblock An aging-independent replicative lifespan in a symmetrically dividing eukaryote.
    \newblock {\em Elife}, 6, January 2017.
    
    \bibitem{Wang2010-td}
    Ping Wang, Lydia Robert, James Pelletier, Wei~Lien Dang, Francois Taddei, Andrew Wright, and Suckjoon Jun.
    \newblock Robust growth of escherichia coli.
    \newblock {\em Curr. Biol.}, 20(12):1099--1103, June 2010.
    
    \bibitem{Allocati2015-oo}
    N~Allocati, M~Masulli, C~Di~Ilio, and V~De~Laurenzi.
    \newblock Die for the community: an overview of programmed cell death in bacteria.
    \newblock {\em Cell Death Dis.}, 6(1):e1609, January 2015.
    
    \bibitem{Horowitz2010-rd}
    Joseph Horowitz, Mark~D Normand, Maria~G Corradini, and Micha Peleg.
    \newblock Probabilistic model of microbial cell growth, division, and mortality.
    \newblock {\em Appl. Environ. Microbiol.}, 76(1):230--242, January 2010.
    
    \bibitem{Fagerlind2012-ev}
    Magnus~G Fagerlind, Jeremy~S Webb, Nicolas Barraud, Diane McDougald, Andreas Jansson, Patric Nilsson, Mikael Harl{\'e}n, Staffan Kjelleberg, and Scott~A Rice.
    \newblock Dynamic modelling of cell death during biofilm development.
    \newblock {\em J. Theor. Biol.}, 295:23--36, February 2012.
    
    \bibitem{Himeoka2020-ho}
    Y~Himeoka and N~Mitarai.
    \newblock Dynamics of bacterial populations under the feast-famine cycles.
    \newblock {\em Physical Review Research}, 2020.
    
    \bibitem{Schink2019-dd}
    Severin~J Schink, Elena Biselli, Constantin Ammar, and Ulrich Gerland.
    \newblock Death rate of e. coli during starvation is set by maintenance cost and biomass recycling.
    \newblock {\em Cell Syst}, 9(1):64--73.e3, July 2019.
    
    \bibitem{Biselli2020-fs}
    Elena Biselli, Severin~Josef Schink, and Ulrich Gerland.
    \newblock Slower growth of escherichia coli leads to longer survival in carbon starvation due to a decrease in the maintenance rate.
    \newblock {\em Mol. Syst. Biol.}, 16(6):e9478, June 2020.
    
    \bibitem{Maire2020-ty}
    Th{\'e}o Maire, Tim Allertz, Max~A Betjes, and Hyun Youk.
    \newblock Dormancy-to-death transition in yeast spores occurs due to gradual loss of gene-expressing ability.
    \newblock {\em Mol. Syst. Biol.}, 16(11):e9245, November 2020.
    
    \bibitem{Wu2024-vz}
    Renfei Wu, Cong Li, Jiuyi Li, Jelmer Sjollema, G^^c3^^a9sinda~I Geertsema-Doornbusch, H~Willy de~Haan-Visser, Emma S~C Dijkstra, Yijin Ren, Zexin Zhang, Jian Liu, Hans~C Flemming, Henk~J Busscher, and Henny~C van~der Mei.
    \newblock Bacterial killing and the dimensions of bacterial death.
    \newblock {\em NPJ Biofilms Microbiomes}, 10(1):87, September 2024.
    
    \bibitem{Schink2021-vp}
    S~J Schink, M~Polk, E~Athaide, A~Mukherjee, C~Ammar, and {others}.
    \newblock The energy requirements of ion homeostasis determine the lifespan of starving bacteria.
    \newblock {\em bioRxiv}, 2021.
    
    \bibitem{Gray2019-kr}
    Declan~A Gray, Gaurav Dugar, Pamela Gamba, Henrik Strahl, Martijs~J Jonker, and Leendert~W Hamoen.
    \newblock Extreme slow growth as alternative strategy to survive deep starvation in bacteria.
    \newblock {\em Nat. Commun.}, 10(1):1--12, February 2019.
    
    \bibitem{Umetani2021-vz}
    M~Umetani, M~Fujisawa, R~Okura, T~Nozoe, S~Suenaga, H~Nakaoka, E~Kussell, and Y~Wakamoto.
    \newblock Observation of non-dormant persister cells reveals diverse modes of survival in antibiotic persistence.
    \newblock 2021.
    
    \bibitem{Laman_Trip2022-so}
    Diederik~S Laman~Trip, Th{\'e}o Maire, and Hyun Youk.
    \newblock Slowest possible replicative life at frigid temperatures for yeast.
    \newblock {\em Nat. Commun.}, 13(1):7518, December 2022.
    
    \bibitem{Walker2023-fd}
    Rachel~M Walker, Valeria~C Sanabria, and Hyun Youk.
    \newblock Microbial life in slow and stopped lanes.
    \newblock {\em Trends Microbiol.}, December 2023.
    
    \bibitem{Oliver2010-dp}
    James~D Oliver.
    \newblock Recent findings on the viable but nonculturable state in pathogenic bacteria.
    \newblock {\em FEMS Microbiol. Rev.}, 34(4):415--425, July 2010.
    
    \bibitem{xu1982survival}
    Huai~Shu Xu, N~Roberts, FL~Singleton, RW~Attwell, Darrell~Jay Grimes, and RR~Colwell.
    \newblock Survival and viability of nonculturable escherichia coli and vibrio cholerae in the estuarine and marine environment.
    \newblock {\em Microbial ecology}, 8:313--323, 1982.
    
    \bibitem{Song2021-xv}
    Sooyeon Song and Thomas~K Wood.
    \newblock 'viable but non-culturable cells' are dead.
    \newblock {\em Environ. Microbiol.}, 23(5):2335--2338, May 2021.
    
    \bibitem{Kirschner2021-az}
    Alexander K~T Kirschner, Julia Vierheilig, Hans-Curt Flemming, Jost Wingender, and Andreas~H Farnleitner.
    \newblock How dead is dead? viable but non-culturable versus persister cells.
    \newblock {\em Environ. Microbiol. Rep.}, 13(3), 2021.
    
    \bibitem{Jung2019-xl}
    Sung-Hee Jung, Choong-Min Ryu, and Jun-Seob Kim.
    \newblock Bacterial persistence: Fundamentals and clinical importance.
    \newblock {\em J. Microbiol.}, 57(10):829--835, October 2019.
    
    \bibitem{Amato2014-xi}
    Stephanie~M Amato, Christopher~H Fazen, Theresa~C Henry, Wendy W~K Mok, Mehmet~A Orman, Elizabeth~L Sandvik, Katherine~G Volzing, and Mark~P Brynildsen.
    \newblock The role of metabolism in bacterial persistence.
    \newblock {\em Front. Microbiol.}, 5:70, March 2014.
    
    \bibitem{Balaban2011-rx}
    N~Q Balaban.
    \newblock Persistence: mechanisms for triggering and enhancing phenotypic variability.
    \newblock {\em Curr. Opin. Genet. Dev.}, 21(6):768--775, December 2011.
    
    \bibitem{Balaban2004-da}
    Nathalie~Q Balaban, Jack Merrin, Remy Chait, Lukasz Kowalik, and Stanislas Leibler.
    \newblock Bacterial persistence as a phenotypic switch.
    \newblock {\em Science}, 305(5690):1622--1625, September 2004.
    
    \bibitem{Balaban2019-zp}
    Nathalie~Q Balaban, Sophie Helaine, Kim Lewis, Martin Ackermann, Bree Aldridge, Dan~I Andersson, Mark~P Brynildsen, Dirk Bumann, Andrew Camilli, James~J Collins, Christoph Dehio, Sarah Fortune, Jean-Marc Ghigo, Wolf-Dietrich Hardt, Alexander Harms, Matthias Heinemann, Deborah~T Hung, Urs Jenal, Bruce~R Levin, Jan Michiels, Gisela Storz, Man-Wah Tan, Tanel Tenson, Laurence Van~Melderen, and Annelies Zinkernagel.
    \newblock Definitions and guidelines for research on antibiotic persistence.
    \newblock {\em Nat. Rev. Microbiol.}, 17(7):441--448, July 2019.
    
    \bibitem{Barkai1997-by}
    N~Barkai and S~Leibler.
    \newblock Robustness in simple biochemical networks.
    \newblock {\em Nature}, 387(6636):913--917, June 1997.
    
    \bibitem{Berg2004-aa}
    H~C Berg.
    \newblock Rotary motor.
    \newblock In Howard~C Berg, editor, {\em E. coli in Motion}, pages 105--120. Springer New York, New York, NY, 2004.
    
    \bibitem{Shinar2007-vj}
    Guy Shinar, Ron Milo, Mar{\'\i}a~Rodr{\'\i}guez Mart{\'\i}nez, and Uri Alon.
    \newblock Input--output robustness in simple bacterial signaling systems.
    \newblock {\em Proceedings of the National Academy of Sciences}, 104(50):19931--19935, 2007.
    
    \bibitem{Shinar2010-ns}
    Guy Shinar and Martin Feinberg.
    \newblock Structural sources of robustness in biochemical reaction networks.
    \newblock {\em Science}, 327(5971):1389--1391, March 2010.
    
    \bibitem{Hopfield1974-tj}
    J~J Hopfield.
    \newblock Kinetic proofreading: a new mechanism for reducing errors in biosynthetic processes requiring high specificity.
    \newblock {\em Proc. Natl. Acad. Sci. U. S. A.}, 71(10):4135--4139, October 1974.
    
    \bibitem{Thornburg2022-nm}
    Zane~R Thornburg, David~M Bianchi, Troy~A Brier, Benjamin~R Gilbert, Tyler~M Earnest, Marcelo C~R Melo, Nataliya Safronova, James~P S{\'a}enz, Andr{\'a}s~T Cook, Kim~S Wise, Clyde~A Hutchison, 3rd, Hamilton~O Smith, John~I Glass, and Zaida Luthey-Schulten.
    \newblock Fundamental behaviors emerge from simulations of a living minimal cell.
    \newblock {\em Cell}, 185(2):345--360.e28, January 2022.
    
    \bibitem{Li2022-cq}
    Feiran Li, Le~Yuan, Hongzhong Lu, Gang Li, Yu~Chen, Martin K~M Engqvist, Eduard~J Kerkhoven, and Jens Nielsen.
    \newblock Deep learning-based kcat prediction enables improved enzyme-constrained model reconstruction.
    \newblock {\em Nature Catalysis}, 5(8):662--672, June 2022.
    
    \bibitem{Choudhury2022-vv}
    Subham Choudhury, Michael Moret, Pierre Salvy, Daniel Weilandt, Vassily Hatzimanikatis, and Ljubisa Miskovic.
    \newblock Reconstructing kinetic models for dynamical studies of metabolism using generative adversarial networks.
    \newblock {\em Nat Mach Intell}, 4(8):710--719, August 2022.
    
    \bibitem{Himeoka2022-dh}
    Y~Himeoka and N~Mitarai.
    \newblock Emergence of growth and dormancy from a kinetic model of the escherichia coli central carbon metabolism.
    \newblock {\em Physical Review Research}, 2022.
    
    \bibitem{stone2008death}
    Jacqueline~I Stone and Mariko~Namba Walter.
    \newblock {\em Death and the afterlife in Japanese Buddhism}.
    \newblock University of Hawaii Press, 2008.
    
    \bibitem{wolfenden2001depth}
    Richard Wolfenden and Mark~J Snider.
    \newblock The depth of chemical time and the power of enzymes as catalysts.
    \newblock {\em Accounts of chemical research}, 34(12):938--945, 2001.
    
    \bibitem{Saperstone1973-zk}
    Stephen~H Saperstone.
    \newblock Global controllability of linear systems with positive controls.
    \newblock {\em SIAM J. Control Optim.}, 11(3):417--423, August 1973.
    
    \bibitem{Farkas1998-zk}
    Gyula Farkas.
    \newblock Local controllability of reactions.
    \newblock {\em J. Math. Chem.}, 24(1):1--14, August 1998.
    
    \bibitem{Dochain1992-gh}
    D~Dochain and L~Chen.
    \newblock Local observability and controllability of stirred tank reactors.
    \newblock {\em J. Process Control}, 2(3):139--144, January 1992.
    
    \bibitem{Drexler2016-bu}
    D{\'a}niel~Andr{\'a}s Drexler and J{\'a}nos T{\'o}th.
    \newblock Global controllability of chemical reactions.
    \newblock {\em J. Math. Chem.}, 54(6):1327--1350, June 2016.
    
    \bibitem{atkins2023atkins}
    Peter~William Atkins, Julio De~Paula, and James Keeler.
    \newblock {\em Atkins' physical chemistry}.
    \newblock Oxford university press, 2023.
    
    \bibitem{Feinberg2019-hp}
    Martin Feinberg.
    \newblock {\em Foundations of Chemical Reaction Network Theory}.
    \newblock Springer, Cham, 2019.
    
    \bibitem{cornish2013fundamentals}
    A.~Cornish-Bowden.
    \newblock {\em Fundamentals of Enzyme Kinetics}.
    \newblock Wiley, 2013.
    
    \bibitem{nijmeijer1990nonlinear}
    Henk Nijmeijer and Arjan Van~der Schaft.
    \newblock {\em Nonlinear dynamical control systems}, volume 464.
    \newblock Springer, 1990.
    
    \bibitem{coron2007control}
    Jean-Michel Coron.
    \newblock {\em Control and nonlinearity}.
    \newblock Number 136. American Mathematical Soc., 2007.
    
    \bibitem{collins2004transmissible}
    Steven~J Collins, Victoria~A Lawson, and Colin~L Masters.
    \newblock Transmissible spongiform encephalopathies.
    \newblock {\em The Lancet}, 363(9402):51--61, 2004.
    
    \bibitem{schrodinger1944life}
    Erwin Schr{\"o}dinger.
    \newblock {\em What is life? The physical aspect of the living cell and mind}.
    \newblock Cambridge university press Cambridge, 1944.
    
    \bibitem{pirt1965maintenance}
    SJ~Pirt.
    \newblock The maintenance energy of bacteria in growing cultures.
    \newblock {\em Proceedings of the Royal Society of London. Series B. Biological Sciences}, 163(991):224--231, 1965.
    
    \bibitem{gouy2016results}
    S~Gouy, G~Ferron, O~Glehen, A~Bayar, F~Marchal, C~Pomel, F~Quenet, JM~Bereder, MC~Le~Deley, and P~Morice.
    \newblock Results of a multicenter phase i dose-finding trial of hyperthermic intraperitoneal cisplatin after neoadjuvant chemotherapy and complete cytoreductive surgery and followed by maintenance bevacizumab in initially unresectable ovarian cancer.
    \newblock {\em Gynecologic Oncology}, 142(2):237--242, 2016.
    
    \bibitem{ciesielski2022erebosis}
    Hanna~M Ciesielski, Hiroshi Nishida, Tomomi Takano, Aya Fukuhara, Tetsuhisa Otani, Yuko Ikegawa, Morihiro Okada, Takashi Nishimura, Mikio Furuse, and Sa~Kan Yoo.
    \newblock Erebosis, a new cell death mechanism during homeostatic turnover of gut enterocytes.
    \newblock {\em PLoS Biology}, 20(4):e3001586, 2022.
    
    \bibitem{D-Arcy2019-an}
    Mark~S D'Arcy.
    \newblock Cell death: a review of the major forms of apoptosis, necrosis and autophagy.
    \newblock {\em Cell Biol. Int.}, 43(6):582--592, June 2019.
    
    \bibitem{Newton2024-hh}
    Kim Newton, Andreas Strasser, Nobuhiko Kayagaki, and Vishva~M Dixit.
    \newblock Cell death.
    \newblock {\em Cell}, 187(2):235--256, January 2024.
    
    \bibitem{Yuan2024-nv}
    Junying Yuan and Dimitry Ofengeim.
    \newblock A guide to cell death pathways.
    \newblock {\em Nat. Rev. Mol. Cell Biol.}, 25(5):379--395, May 2024.
    
    \bibitem{Park2023-zw}
    Wonyoung Park, Shibo Wei, Bo-Sung Kim, Bosung Kim, Sung-Jin Bae, Young~Chan Chae, Dongryeol Ryu, and Ki-Tae Ha.
    \newblock Diversity and complexity of cell death: a historical review.
    \newblock {\em Exp. Mol. Med.}, 55(8):1573--1594, August 2023.
    
    \end{thebibliography}
\end{document}